\documentclass[journal]{IEEEtran}

\usepackage{hyperref}
\hypersetup{ colorlinks=true, linkcolor=blue, filecolor=magenta, urlcolor=cyan,}
\usepackage{cite}

\ifCLASSINFOpdf
\else
\fi
\usepackage{amsthm}

\usepackage{amsmath}
\usepackage{mathrsfs}
\usepackage{amsfonts} 
\usepackage{amssymb}
\usepackage{graphicx}
\usepackage{listings}
\usepackage{multirow}
\usepackage{array, makecell}
\usepackage[mathscr]{euscript}
\usepackage{threeparttable}

\begin{document}

%
% paper title
% Titles are generally capitalized except for words such as a, an, and, as,
% at, but, by, for, in, nor, of, on, or, the, to and up, which are usually
% not capitalized unless they are the first or last word of the title.
% Linebreaks \\ can be used within to get better formatting as desired.
% Do not put math or special symbols in the title.
\title{Fast and Optimal Adaptive Tracking Control: \\A Novel Meta-Reinforcement Learning via Conditional Generative Adversarial Net}

\author{Mohammad~Mahmoudi,
            Nasser~Sadati, \IEEEmembership{Senior Member,~IEEE}% <-this % stops a space           
\thanks{M. Mahmoudi and N. Sadati are with the Department of Electrical Engineering, Sharif University of Technology, Tehran, Iran (e-mails: m.mahmoudi.f75@gmail.com, sadati@sharif.edu). }}%
% The paper headers
\markboth{Submitted to IEEE Transactions on}%
{Mahmoudi \MakeLowercase{\textit{et al.}}:Fast and Optimal Adaptive Tracking Control: A Novel Meta-Reinforcement Learning via Conditional Generative Adversarial Net}%
% make the title area
\maketitle

% As a general rule, do not put math, special symbols or citations
% in the abstract or keywords.
\begin{abstract}

The control of nonlinear systems with unknown dynamics has been a significant field of research for many years.
This paper presents a novel data-driven optimal adaptive control structure with less control effort and faster adaptation than standard adaptive control counterparts.
The proposed control structure utilizes the system's recorded data to increase the speed of adaptation and performance dramatically.
In this study, we employ a conditional generative adversarial net (CGAN) as a novel central pattern generator to reproduce the steady-state harmonic pattern of the control signals matching the system's uncertainties over a wide range.
We can also use the CGAN architecture as a fault detector.
The CGAN provides a low-dimensional latent space of uncertainties. It enables rapid and convenient adaptation when there are many parametric uncertainties, especially for large-scale systems. 
Then, we introduce a novel meta-reinforcement learning framework to adapt the latent space of CGAN to the
system's uncertainties as an optimal direct adaptive controller without any system identifier. 
Another part of the control structure is a regulator that achieves semi-global asymptotic tracking using the Lyapunov stability analysis. 
Finally, via some simulations, we evaluate the capabilities of the proposed designs on two dynamical systems, a robot manipulator and a large-scale musculoskeletal structure, in the presence of disturbance and perturbation. 
\end{abstract}

% Note that keywords are not normally used for peerreview papers.
\begin{IEEEkeywords}
Optimal adaptive control, central pattern generator (CPG), conditional generative adversarial net (CGAN), data-driven control, meta-reinforcement learning (Meta-RL). 
\end{IEEEkeywords}

% For peer review papers, you can put extra information on the cover
% page as needed:
% \ifCLASSOPTIONpeerreview
% \begin{center} \bfseries EDICS Category: 3-BBND \end{center}
% \fi
%
% For peerreview papers, this IEEEtran command inserts a page break and
% creates the second title. It will be ignored for other modes.
\IEEEpeerreviewmaketitle

\section{Introduction}

\newtheorem{theorem}{Theorem}
\newtheorem{lemma}{Lemma}
\newtheorem{assumption}{Assumption}
\newtheorem{remark}{Remark}

\IEEEPARstart{N}{umerous} physical systems and phenomena are intrinsically nonlinear, and there are many structured and unstructured uncertainties in their dynamic models
\cite{Slotine1991}. Therefore, developing efficient control methods for such systems is of great importance for researchers. Various adaptive and robust controllers have been presented by
\cite{Slotine1991, chen2015, Astrom1994}  
that can be used for control of uncertain nonlinear systems in tracking tasks. 
Almost all require knowledge of the system dynamic equations for stability analysis and designing adaption and control law. 
Consequently, approximation-based adaptive control approaches have been suggested to compensate for the unknown dynamics by employing an approximator (e.g., neural network (NN), fuzzy system, neuro-fuzzy network)
\cite{hagan1999, Jang1997, sarangapani2006, ge2004, Hayakawa2008, ge2002,le2022, Patre2010, Patre2011, Dierks2009, Sharma2012a, Sharma2012b, Shin2012, Yang2015 }.
While some leverage an approximator to estimate dynamic systems as an indirect adaptive control method, others use it directly to approximate the controller as a direct adaptive control method.  

There are different methods to update the NN weights in the NN-based adaptive control. 
Gradient-based methods, reviewed by \cite{hagan1999}, estimate the Jacobin of the system through a NN. In gradient-free methods, the NN weights are chosen by a derivative-free optimization method with high computational costs \cite{Jang1997}. 
These two categories of methods generally do not guarantee the stability of the closed-loop system.
In recent years, Lyapunov-based methods have been suggested. 
They employ Lyapunov stability analysis to drive adaptation laws for updating weights adaptively and stably 
\cite{sarangapani2006, ge2004, ge2002,le2022,Hayakawa2008, Patre2010, Patre2011, Dierks2009, Sharma2012a, Sharma2012b, Shin2012, Yang2015 }.
Some Lyapunov-based methods ensure that the closed-loop system is uniformly ultimately bounded, and indeed the ultimate bound on the tracking error depends on the reconstruction error of the NN
\cite{sarangapani2006, ge2004}. 
To achieve asymptotic tracking stability, reference
\cite{Hayakawa2008}
proposed an approach for a specific form of reconstruction error, various others 
\cite{ge2002, le2022, Patre2010, Patre2011, Dierks2009, Sharma2012a, Sharma2012b, Shin2012, Yang2015}
suggested adding a variable structure control (VSC) term in the control law.
As a VSC, a few of these researchers 
\cite{ge2002,le2022}
used the sign function of error, whereas the rest 
\cite{Patre2010, Patre2011, Dierks2009, Sharma2012a, Sharma2012b, Shin2012, Yang2015}
used the robust integral of the sign of error (RISE)
\cite{Xian2004}
to eliminate the ultimate bound dependent on the reconstruction error of the NN. The combination RISE and NN is also known as the RISE-NN structure. 

The RISE-NN consists of a Lyapunov-based dynamic NN to estimate the uncertain system dynamics incorporated with the RISE controller as a high-gain controller. The NN and RISE, respectively, play the feedforward and feedback control roles in this arrangement, yielding semi-global asymptotic tracking stability
\cite{Patre2010, Patre2011, Dierks2009, Sharma2012a, Sharma2012b, Shin2012, Yang2015 }. In the RISE-NN structure,
the RISE controller covers the poor initial functioning of the NN.
In other words, the high-gain controller compensates for more uncertainties instead of NN, and control system performance is reduced. Hence, these works do not present any evaluation of the capability of NN modeling. 
As experience shows, in the structures mentioned above, if the system's uncertainty changes over time, their performances will decrease. Also, they cannot provide optimal performance.

Controllers based on reinforcement learning (RL) are the other adaptive control methods. 
RL-based controllers are bio-inspired and attempt to model the learning mechanism in animals and humans and update the control system's parameters according to the controller's interaction with the dynamic plant
\cite{KHAN2012}.
In recent years, many policy-gradient algorithms have been introduced 
\cite{Schulman2015, Silver2014 ,Lillicrap2016a}
to control a Markov decision process (MDP) as model-free RL approaches that must interact with the system to a considerable degree to converge to a suitable controller
\cite{Nagabandi2019}. 
Hence, model-free RL faces many challenges and difficulties in the learning process and is inappropriate for real-world systems, particularly complex nonlinear systems with continuous action space.
Meanwhile, model-based approaches obtain the optimal policy by estimating the dynamic model simultaneously. There have been a variety of researches for tracking tasks as model-based RL 
\cite{Hu2020 , Modares2013, Kamalapurkar2017},
a kernel-based dynamic model was investigated by 
\cite{Hu2020}, and 
references \cite{Modares2013, Kamalapurkar2017}
employed the system identification techniques. 
Lately, some researchers 
\cite{Finn2017, Rakelly2019, Nagabandi2019}
have introduced the Meta-RL framework, a powerful concept for improving learning efficiency and speeding up online adaptation.
Therefore, recent developments demonstrate that using the RL framework is advantageous in designing controller structures owing to optimality and adaptivity properties. 

In this paper, we embed the conditional generative adversarial net (CGAN) architecture
\cite{Mirza2014} 
into the RISE-NN structure incorporated into the Meta-RL framework to provide a novel Meta-RL approach for fast and optimal adaptive tracking control. 
Both NN and RISE controllers can collaborate on compensating for uncertainties
in the RISE-NN structures described in earlier studies 
\cite{Patre2010, Patre2011, Dierks2009, Sharma2012a, Sharma2012b, Shin2012, Yang2015}. As a result, the more NN approximation error, the more RISE collaboration, and consequently, high control gain is required to integrate the sign of error term 
\cite{Yang2015}.
Eventually, the control signal should have wide bandwidth and large amplitude chattering, and it may excite the high-order unmodeled dynamics. For this reason, the capability of NN modeling is essential for control system performance. To tackle this issue, we employ the generative adversarial net (GAN) mechanism
\cite{Goodfellow2014} 
instead of the simple NN based on the solid aptitude of the GAN that can model complicated data distributions. The GAN also has many advantages compared with its counterparts
\cite{Creswell2018}. 
This paper is the first study to exploit the strength of the CGAN to build a powerful adaptive control. 
Moreover, we demonstrate that the CGAN architecture is an outstanding adaptable central pattern generator (CPG) for robotic applications
\cite{Ijspeert2008}.
Another significant finding is using CGAN architecture as a data-driven fault detector \cite{ding2014}.
According to
\cite{Chowdhary2010, Adam2012},
using the system's previous data and a experience replay buffer technique is valuable to satisfy the persistent excitation (PE) condition and improves data efficiency in adaptive systems and online RL algorithms. 
In a way, we use the distillation technique and the previous data for training the CGAN and leverage a experience replay buffer in the proposed framework.
The CGAN generates the required harmonic pattern for exciting the system and leads to a dynamic policy for steady-state control. Its latent space represents the system's uncertainties in a low-dimensional space.
It is necessary to match the latent space of the CGAN to the system's uncertainties adaptively. Therefore, we have developed two adaptation mechanisms to adjust the input noise of the CGAN. In the first approach, a Meta-RL framework is proposed as a fast and optimal adaptive control using deep deterministic policy gradient (DDPG)
\cite{Lillicrap2016a} and regarded as the direct approach.
In the second, we offer a model-based approach to compare with the first by using the  extended Kalman filter (EKF) as a system identifier, referred to as the indirect approach.
The RISE controller plays the role of the regulator control by compensating for transient errors; additionally, it ensures that the system is stable and robust to any perturbations, including the reconstruction error of the NN, unmodeled dynamics, and disturbance. In summary, it takes care of the controlled system. We also use Lyapunov stability analysis to show semi-global asymptotic tracking. Finally, simulation results on a planar two-degree-of-freedom (DoF) manipulator illustrate the impact of CGAN and RL on control system performance compared to the previous method \cite{Yang2015}. In addition, we apply the proposed controller to a large-scale musculoskeletal system to demonstrate the scalability of our design.

Therefore, the contributions of this paper are summarized as follows:
\begin{itemize}
\item Proposing a new Meta-RL framework to obtain a fast and optimal adaptive control by embedding CGAN and RL into the RISE structure to ensure asymptotic stability and robustness.

\item Introducing an innovative data-driven adaptable CPG by using CGAN for applications in robot control.

\item Providing a novel data-driven fault detector by using CGAN architecture.
\end{itemize}
We organize the remainder of the paper as follows.
In Section \ref{section: Problem Formulation}, we describe the problem formulation.
Section \ref{section: Proposed Control Schemes} presents the proposed control structures.
In Section \ref{section: Stability Analysis}, we analyze the stability of the proposed control schemes. 
Section \ref{section: Simulation Results} is devoted to the simulation and analysis results.
Finally, Section \ref{section: Conclusion and Future work} concludes the paper.

%%%%%%%%%%%%%%%%%%%%%%%%%%%%%%%%%%%%%%%%%%%%%%%%%%%%%%%%%%%%%%%%%%%%%%%%%%%%%%%%%%%%%%%%%%%%%%%%%%%%%%
\section{Problem Formulation} \label{section: Problem Formulation}
In this paper, we consider a class of $(nm)$th-order, multi-input-multi-output (MIMO) uncertain nonlinear input-affine continuous-time systems described by
\begin{align}\label{eq2.1}
x^{(m)}(t)&=f(x(t),\dot{x}(t),...,x^{(m-1)}(t)) \nonumber \\ 
&+g(x(t),\dot{x}(t),...,x^{(m-1)}(t)) u(t)+d(t) \nonumber \\
y(t) &= x(t),
\end{align}
where $(.)^{(i)}(t)$ indicates the $i^{\text{th}}$ derivative w.r.t. time; $x^{(i)}(t)\in \mathbb{R}^n$, $i=0,\,...,\,m-1$ are the measurable system states, which can be written in a vector form as $\bar{x}(t) \triangleq \left[x(t)^T,\,\dot{x}(t)^T,\,...,\,x^{(m-1)}(t)^T\right]^T \in \mathbb{R}^{nm}$ that denotes the state vector; and $u(t)\in \mathbb{R}^n$, $y(t)\in \mathbb{R}^n$, and $d(t)\in \mathbb{R}^n$ are the control input vector, the output vector, and the unknown bounded disturbance or unmodeled dynamics vector, respectively. Let $f(.)\in \mathbb{R}^n$ and $g(.)\in \mathbb{R}^{n\times n}$ represent the uncertain second-order globally differentiable nonlinear functions. 
Designing optimal adaptive control such that $x(t)$ (as the system output) asymptotically tracks the desired trajectory $x_d(t)\in \mathbb{R}^n$ in the presence of uncertainties and disturbance is the ultimate goal of this research. 
For subsequent analyses, we consider the following assumptions, which have also been exploited in the literature \cite{Patre2010, Patre2011, Dierks2009, Yang2015,Xian2004}.

\begin{assumption} \label{Assumption_G}
The function $g(.)$ is a positive-definite matrix function such that 
\begin{equation}
\underline{g}\|\eta\|^2 \leq \eta^Tg(\bar{x})\eta \leq \bar{g}(\bar{x}) \|\eta\|^2 \indent\forall \eta \in \mathbb{R}^n,\forall \bar{x} \in \mathbb{R}^{nm},
\end{equation}
where $\underline{g}\in \mathbb{R}$ is a positive constant and $\bar{g}(\bar{x})\in \mathbb{R}$ is a positive non-decreasing function, and $\|.\|$ stands for the standard Euclidean norm.
\end{assumption}

\begin{assumption} \label{Assumption_desired}
The desired trajectory $x_d(t)$ is $(m+2)$th-order differentiable w.r.t. time. Also, $x_d(t) $ and all its time derivatives are bounded as follows:
\begin{equation}
\|x_d(t)\|\leq c_{x_{d0}}, \|\dot{x}_d(t)\|\leq c_{x_{d1}},...,\|x^{(m+2)}_d(t)\|\leq c_{x_{d(m+2)}},
\end{equation}
where $c_{x_{d0}} , c_{x_{d1}},...,$ and $c_{x_{d(m+2)}} \in \mathbb{R}^+$ are unknown positive constants.
\end{assumption}

\begin{assumption} \label{Assumption_desired_func} 
There exists a locally Lipschitz function $h:\mathbb{R}^{mn}\rightarrow\mathbb{R}^n$ such that
\begin{equation}\label{eq2.4}
x^{(m)}_d(t)=h(x_d(t),\dot{x}_d(t),...,x^{(m-1)}_d(t)).
\end{equation}
\end{assumption}

\begin{assumption} \label{Assumption_disturbance} 
The disturbance $d(t)$ is bounded and has second-order, bounded derivatives w.r.t. time such that
\begin{equation}
\|d(t)\|\leq c_{d_{0}},\|\dot{d}(t)\|\leq c_{d_{1}},\|\ddot{d}(t)\|\leq c_{d_{2}},
\end{equation}
where $c_{d_{0}} , c_{d_{1}},$ and $c_{d_{2}} \in \mathbb{R}^+$ are unknown positive constants.
\end{assumption}

\begin{remark} 
Under Assumption \ref{Assumption_G}, it is clearly visible that the system \eqref{eq2.1} is full-state feedback linearizable, and according to Theorem 7.3.41 of \cite{vidyasagar2002}, the reachability of the system \eqref{eq2.1} can be shown.
\end{remark}

%%%%%%%%%%%%%%%%%%%%%%%%%%%%%%%%%%%%%%%%%%%%%%%%%%%%%%%%%%%%%%%%%%%%%%%%%%%%%%%%%%%%%%%%%%%%%%%%%%%%%%

\section{Proposed Control Schemes} \label{section: Proposed Control Schemes}
Our objective is to ensure that the system output tracks a desired trajectory. We define tracking error, shown by $\tilde{x}(t)$, as follows:
\begin{equation}
\tilde{x}(t)\triangleq x(t)-x_d(t).
\end{equation}
Also, the open-loop tracking error dynamics, without disturbance, can be represented as 
\begin{equation}\label{eq3.7}
\tilde{x}^{(m)}(t) = f(\bar{x}(t))+g(\bar{x}(t))u(t)-{x}^{(m)}_d(t).
\end{equation}
The error dynamics \eqref{eq3.7} are nonautonomous. To make the control design more convenient and interpretable, we exploit the method described in earlier studies 
\cite{Kiumarsi2018,  Kamalapurkar2017, KAMALAPURKAR2015}
to formulate the problem autonomously. To do so, by substituting the desired trajectory for $\bar{x}$ in the dynamic system \eqref{eq3.7}, according to Assumption \ref{Assumption_G}, we can obtain the desired steady-state control $u_{ss}(x_d)\in\mathbb{R}^n$ as 
\begin{equation}\label{eq3.8}
u_{ss}(\bar{x}_d) \triangleq g(\bar{x}_d)^{-1}(x^{(m)}_d-f(\bar{x}_d)).
\end{equation}
We define an extended state vector $X\in\mathbb{R}^{2nm}$ as
\begin{equation}\label{eq3.9}
X(t) \triangleq \left[\tilde{x}^T, \dot{\tilde{x}}^T,...,\tilde{x}^{{(m-1)}^T},x^T_d,\dot{x}_d^T,...,x^{{(m-1)}^T}_d\right]^T.
\end{equation}
By substituting \eqref{eq2.4} into \eqref{eq3.7} and by adding and subtracting $g(\bar{x}(t))u_{ss}(t)$ to the result with some manipulation, the extended tracking error dynamics, based on the extended state vector, can be rewritten as:
\begin{align}\label{eq3.10}
\begin{bmatrix}
\tilde{x}^{(m)} \\ 
x^{(m)}_d
\end{bmatrix}
&=
\underbrace{
\begin{bmatrix}
f(\bar{\tilde{x}}+\bar{x}_d)+g(\bar{\tilde{x}}+\bar{x}_d)u_{ss}(\bar{x}_d)-h(\bar{x}_d)\\
h(\bar{x}_d)
\end{bmatrix}}_{F(X)}  \notag \\ 
&+
\underbrace{\begin{bmatrix}
g(\bar{\tilde{x}}+\bar{x}_d) \\
0_{n\times n}
\end{bmatrix}}_{G(X)}
u_{reg},
\end{align}
where $u_{reg}\in \mathbb{R}^n$ is the regulation control as follows:
\begin{equation}\label{eq3.11}
u_{reg} \triangleq u-u_{ss}(x_d). 
\end{equation}
Because the functions $f$, $g$, and $h$ are locally Lipschitz, we can say that the functions $F$ and $G$ are also locally Lipschitz. $\tilde{x}\equiv 0$ is the equilibrium point of the first subsystem of \eqref{eq3.10}.
The representation of the error dynamics \eqref{eq3.10} is time-invariant.
According to \eqref{eq3.11}, the control input is divided into two parts as $u=u_{ss}+u_{reg}$.
The former is steady-state control, and the latter is regulation control. In some studies \cite{Khalil2002}, they are also known as feedforward and feedback control, respectively. The steady-state control is the inverse dynamics of the system for the desired trajectory (according to \eqref{eq3.8}). The regulation control is the input of the error dynamics \eqref{eq3.10} and must reduce transient error and take the system to the desired steady-state trajectory.
Researchers
\cite{Patre2010, Patre2011, Dierks2009, Sharma2012a, Sharma2012b, Shin2012, Yang2015 }
have leveraged the RISE controller as the regulator and dynamic NN as the steady-state control. Likewise, reference \cite{Zomaya1994} has employed the RL-based controller as the regulator and the inverse dynamic model as the steady-state control without any stability guarantee in an offline adaptive manner. Additionally, investigations have shown that this method has a small region of attraction (RoA), such that stability cannot be ensured if the initial regulation error is significant. 
Our proposed methods applies the RISE controller and an adaptable CGAN, along with two adaptation approaches, as the regulator and the steady-state control, respectively. Therefore, there are two policies in our proposed control structures: 1) CGAN and 2) RISE. We present two adaptive control schemes based on these two policies for the dynamic system \eqref{eq2.1}. 
All of the components of our approaches are shown in Fig. \ref{fig:F1} and described in the following subsections.
\begin{figure}[!h]
  \centering
  \includegraphics[width=0.7\columnwidth]{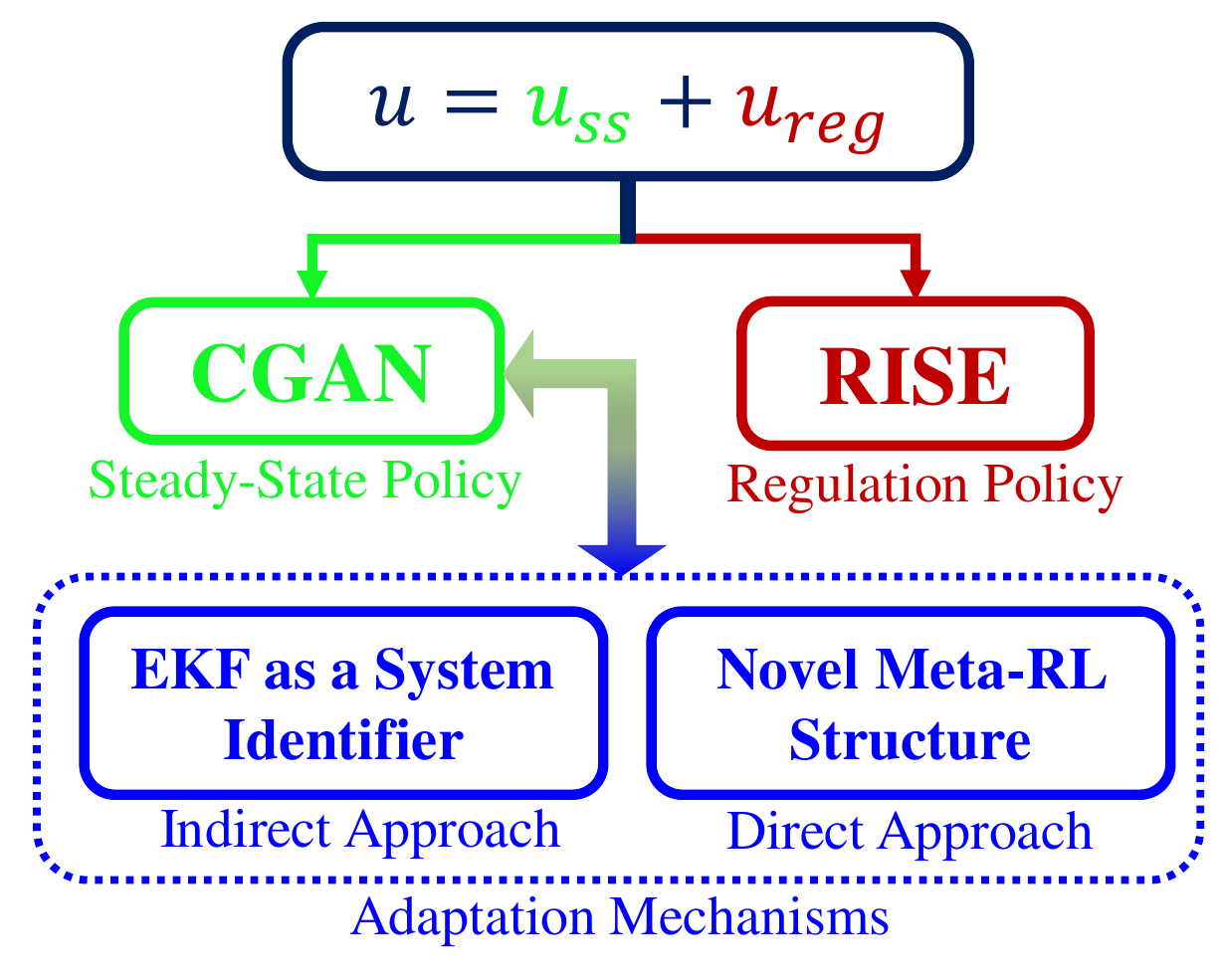}
  \caption{The schematic representation of the proposed control schemes.}\label{fig:F1}
\end{figure}
%%%%%%%%%%%%%%%%%%%%%%%%%%%%%%%%%%%%%%%%%%%%%%%%%%%%%%%%%%%%%%%%%%%%%%%%%%%%%%%%%%%%%%%%%%%%%%%%%%%%%%
\subsection{CGAN: Steady-State Control Policy}
GANs, recently introduced by \cite{Goodfellow2014}, are one of the most fascinating and widely used ideas in machine learning. They have significantly contributed to the progress of artificial intelligence (AI) \cite{Jabbar2021, Creswell2018}. 
GAN is an implicit probabilistic model with outstanding performance in training complex, high-dimensional data distribution.
The generative network transforms a simple distribution into high-dimensional training data distribution. In this framework, another network is a discriminator that trains the generator indirectly in a zero-sum game.
The discriminator and the generator are in a competitive relationship. The discriminator distinguishes the actual data (from the training data set) and the fake data (produced by the generator) and returns a real number between zero and one to indicate the degree to which the sampled data are actual. Concurrently, the generator tries to fool the discriminator by generating data the same as actual data \cite{Goodfellow2014}. 
GANs are used for unsupervised learning, but they are also commonly used for supervised and semi-supervised learning \cite{Rezagholiradeh2018,Olmschenk2019}. The CGAN \cite{Mirza2014} is one of the most famous GAN-based structures for supervised learning. The CGAN aims to generate data conditioned on a given specific label.

As Fig. \ref{fig:F2} shows, this paper uses CGAN to generate the desired steady-state control ($u_{ss}$) by sampling from a low-dimensional multivariate Gaussian (normal) distribution in its input noise ($z$) conditioned on a given specific desired trajectory ($x_d$) as the label.   
Because $u_{ss}$ is the inverse dynamics of the system for $x_d$ and the inverse dynamics has different uncertainties, $u_{ss}$ for a given specific $x_d$ can be varied related to the uncertainties of the dynamic model.
Therefore, using probabilistic models for the generation of $u_{ss}$ is necessary. In general, uncertain inverse dynamics learning is a probabilistic regression task. Reference \cite{Nguyen-tuong2009} used the Gaussian process for uncertain inverse dynamics learning. CGAN, compared to other models, can deal with high-dimensional data and more complicated uncertainties \cite{Rezagholiradeh2018,Olmschenk2019}.
Also, it can represent a low-dimensional latent space of the uncertainties in its input noise and facilitate the adaptation mechanism because searching in the latent space with a lower dimension is more convenient than searching in the original space of the uncertainties.

The required data to train the CGAN for steady-state control are given in $\mathscr{D}_1$, which includes $N$ time-series paired samples of the desired trajectory (label) and the uncertain steady-state control (target) as
\begin{align}
\mathscr{D}_1 = \{label: x^i_d(t),\: target: u^i_{ss}(t)| i=1,...,N\}.
\end{align}

We can obtain the data set $\mathscr{D}_1$ in three ways: 
1) \emph{Using a physics-based inverse dynamic model}: we can generate $u_{ss}$ for some $x_d$ and some parametric and nonparametric uncertainties in the inverse dynamic model based on physical principles. For example, we can model an actuated multi-rigid body \cite{Slotine1991, lewis2003robot} by the Euler-Lagrange equation, then generate $u_{ss}$ for varying mass, size, the inertia of each rigid body, and some nonparametric uncertainties, such as different friction models \cite{Slotine1991} (see Section \ref{section: Simulation Results} for more details); 
2) \emph{Steady-state policy distillation}: according to \eqref{eq3.10} and \eqref{eq3.11}, every controller that can lead the tracking error to zero has the same result for $u_{ss}$. Therefore, we can use a regular base controller that can lead the tracking error to zero for generating $u_{ss}$ for different tasks and uncertainties;
3) \emph{Experimental data in a look-up table}: we can also use the system's empirical data in the steady-state for different tasks.

The desired trajectory in most applications is usually periodic or point-to-point. Also, based on Assumption \ref{Assumption_desired}, a Fourier series (FS) exists for the desired trajectory. 
According to the properties of the dynamic system \eqref{eq2.1}, if the desired trajectory is bounded and periodic, then $u_{ss}$ will also be bounded and periodic. Furthermore, if the fundamental period and frequency of the desired trajectory are $T$ and $\omega$, respectively, then $T$ and $\omega$ can be period and frequency for $u_{ss}$.
Considering the challenges of dealing with time-series signals in the learning tasks, we can use the truncated Fourier series (TFS) coefficients as a feature vector for time-series signals $x_d(t)$ and $u_{ss}(t)$.
We consider a limited number of FS coefficients because high-frequency harmonics are small and negligible. The number of the considered coefficients can be determined by various factors such as the computational power, complexity of the desired trajectory, and the degree of the system's nonlinearity; however, CGAN can handle high-dimensional data. 
Thus, real-valued signals $x_d(t)$ and $u_{ss}(t)$ can be represented by a TFS as follows:
\begin{align}\label{eq3.13}
x_d(t) &= C_{x_d}\Phi^{n_x}_\omega(t)+\delta_{1}({x_d}) \nonumber \\
u_{ss}(t) &= C_{u_{ss}}\Phi^{n_u}_\omega(t)+\delta_{2}({u_{ss}}),
\end{align}
where $C_{x_d}\in \mathbb{R}^{n\times (2n_x+1)}$ and $C_{u_{ss}}\in \mathbb{R}^{n\times (2n_u+1)}$ are TFS coefficients, and $\delta_{1}({x_d})\in \mathbb{R}^{n}$ and $\delta_{2}({u_{ss}})\in \mathbb{R}^{n}$ are the truncation error, for $x_d(t)$ and $u_{ss}(t)$, respectively. $\Phi^{n_{x/u}}_\omega(t)\in \mathbb{R}^{2n_{x/u}+1}$ is sine kernel of the TFS for $x_d(t)$ or $u_{ss}(t)$ as 
\begin{align}\label{eq_Kernel}
\Phi^{n_{x/u}}_\omega(t) = [0.5,&\cos(\omega t),\cos(2\omega t),...,\cos(n_{x/u}\omega t),\nonumber \\
&\sin(\omega t),\sin(2\omega t),...,\sin(n_{x/u}\omega t)]^T.
\end{align}

\begin{figure*}[!h]
  \includegraphics[width=\textwidth]{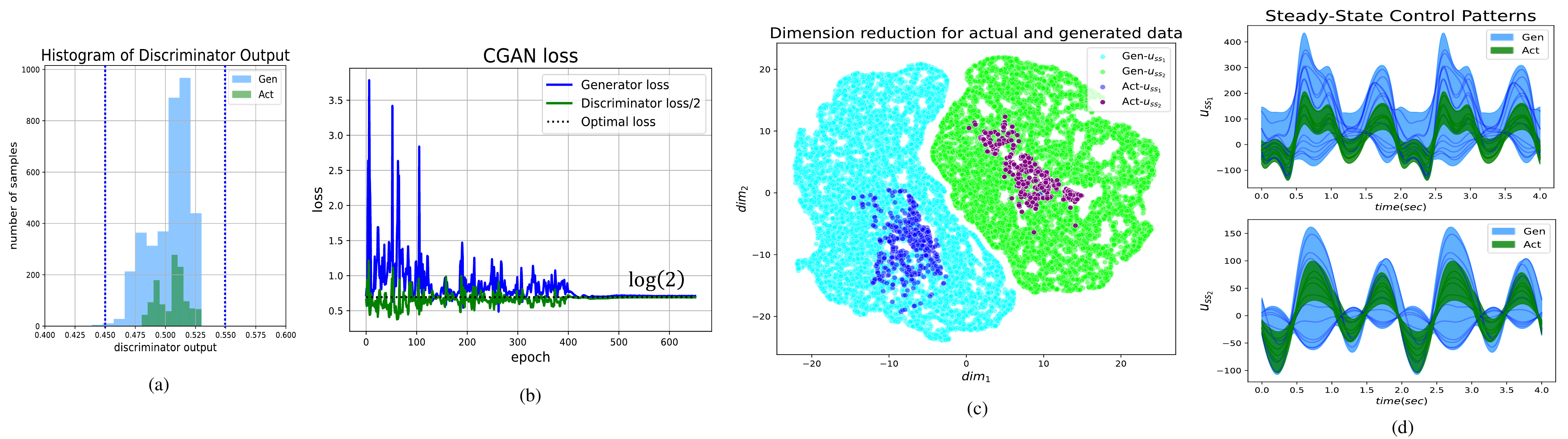}
  \caption{Simulation results of CGAN for the data of a 2-DoF planar manipulator (see Section \ref{section: Exm1} for more details): (a) discriminator output for the actual data (green shade) from the data set and the generated data (blue shade) for a specific label that they are around $\frac{1}{2}$, (b) the loss function value for the generator and the discriminator during training by which they approach their optimal value $\log(2)$, (c) dimension reduction for the actual data distribution (blue for $u_{ss_1}$ and purple for $u_{ss_2}$) from the data set and generated data distribution (cyan for $u_{ss_1}$ and green for  $u_{ss_2}$) for a specific label using the t-distributed stochastic neighbor embedding (t-SNE) method \cite{Maaten2008} to verify the CGAN capability of modeling and generalization, (d) time-series patterns for the actual data (green shade) from the data set and generated data (blue shade) for a specific label.}
\label{fig:F_CGAN}
\end{figure*}

\begin{figure}[!h]
  \centering
  \includegraphics[width=0.8\columnwidth]{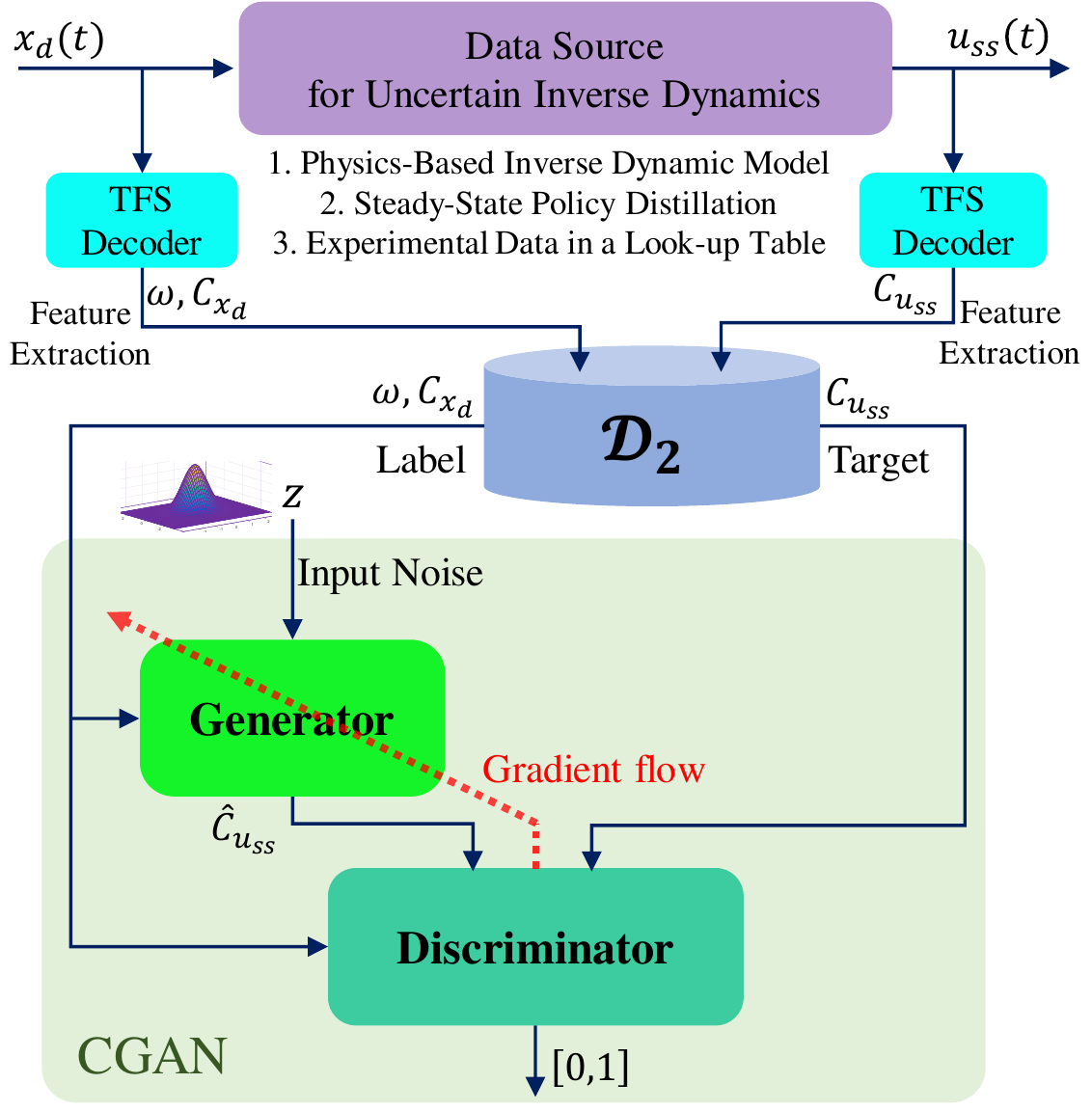}
  \caption{The schematic representation for the training of the CGAN.}\label{fig:F2}
\end{figure}

Therefore, we can use either offline FS transformation \cite{Oppenheim1997} with a moving fixed-length window or the online dynamical canonical system described by \cite{gams2009line} to decode the periodic time-series signals in $\mathscr{D}_1$ into TFS coefficients. Then, we form the data set $\mathscr{D}_2$ and leverage the algorithms given by \cite{Mirza2014} for training the CGAN (Fig. \ref{fig:F2}) as
\begin{align}\label{eq3.15}
\mathscr{D}_2 = \{label: (C^i_{x_d},\omega^i),\:   target: C^i_{u_{ss}}| i=1,...,N\}.
\end{align}

The algorithm's convergence and global optimality have been proven by \cite{Goodfellow2014}. We exhibit the simulation results of training CGAN for a 2-DoF manipulator (see Section \ref{section: Exm1} for more details) in Fig. \ref{fig:F_CGAN}. The global optimality convergence is demonstrated in Fig. \ref{fig:F_CGAN} (a) and (b), and we illustrate the generalization capability of the CGAN in Fig. \ref{fig:F_CGAN} (c) and (d).

Regarding Fig. \ref{fig:F_CGAN} (c) and (d), if the system is in a rather different situation from the training data set, the conditional generator can produce data for this situation owing to the great generalization of the CGAN.
Also, the conditional generator represents a latent space in its input noise space for the uncertainties of the data and dynamic system. It means the conditional generator can generate variant targets corresponding to variant input noise values for a specific desired trajectory (Fig. \ref{fig:F_CGAN} (d)). Hence, to match the uncertainties in the unknown dynamic system, we must adjust the input noise value. As such, two adaptation mechanisms will be presented in the following subsections.
The CGAN trained in the way aforementioned can be exploited for the following applications: 

1) \emph{Steady-state control policy}: 
The conditional generator extracts the inverse dynamic model for steady-state control policy. By treating the input noise $z$ as the adjustable parameters for the steady-state controller, the conditional generator will be a dynamic policy in the Meta-RL framework, as shown in Fig. \ref{fig:F_Generator_Fault} (a). Therefore, the estimation of the steady-state controller ($\hat{u}_{ss}$) is as follows:
\begin{align}\label{eq16}
\hat{u}_{ss}(\bar{x}_d) \triangleq \hat{C}_{u_{ss}}\Phi^{n_u}_\omega(t)=Gen(z(t)|C_{x_d},\omega)\Phi^{n_u}_\omega(t),
\end{align}
where $Gen(.|.)$ is the conditional generator that estimates the TFS coefficients of the steady-state control term. Then, the generated steady-state control $\hat{u}_{ss}$ can be obtained by multiplying generated TFS coefficients by the sine kernel (Fig. \ref{fig:F_Generator_Fault} (a)).

2) \emph{Data-driven describing function}:
From one point of view, the conditional generator is a data-driven extended describing function \cite{Slotine1991} such that it expresses the relationship between the amplitude of the input harmonics and the amplitude of the output harmonics for an uncertain nonlinear system.

3) \emph{Data-driven adaptable CPG}: 
CPGs are biological neural circuits in the neural system of mammalians and invertebrates \cite{Ijspeert2008}. 
CPGs can generate coordinated high-dimensional rhythmic control patterns while receiving low-dimensional input signals from higher-level control centers of the neural system.
The higher-level centers modulate the patterns based on environmental circumstances and the desired frequency \cite{Ijspeert2008, Yu2014}.
In recent years, researchers have been using CPG models for the bio-inspired control of robots with repetitive movements and activities \cite{Yu2014}.  
In this paper, the conditional generator produces TFS coefficients of the rhythmic control patterns while receiving trajectory specifications and low-dimensional input noise to modulate the control patterns according to the system's uncertainties. Therefore, the conditional generator (Fig. \ref{fig:F_Generator_Fault} (a)) is an adaptable CPG model that can be used in the robot control field. 
Previous researchers have suggested some nonlinear dynamic systems as a CPG model that are not flexible enough to adapt to the robot and the environment \cite{Yu2014}. Moreover, adaptive and optimal tuning of the parameters of the previous CPG models is highly challenging. However, we train the conditional generator by the data of the system. It is familiar with the system's behavior and represents a low-dimensional latent space that facilitates adaptation.

4) \emph{Data-driven fault detector}: 
Reference 
\cite{Goodfellow2014} proved after training that the discriminator's optimal value is $\frac{1}{2}$ for all samples from the training set and those generated (Fig. \ref{fig:F_CGAN} (a)). 
In other words, the discriminator does not return $\frac{1}{2}$ for the data that do not have the same distribution as the training set. 
According to Fig. \ref{fig:F_Generator_Fault} (b), we can use the conditional discriminator in a closed-loop robust control system as a fault detector that detects unexpected conditions for the dynamic system.
It takes the TFS coefficients of the control input and returns a real number between zero and one. If the output of the discriminator deviates from $\frac{1}{2}$ in the steady-state, there is a fault in the closed-loop system that means the system does not work similarly to the training data set conditions.
Note that the discriminator can detect the only presence of the fault in the steady-state. It cannot determine the specific sort of fault and its nature. 

Fig. \ref{fig:F_Fault_Curve} shows the fault detection simulation result for a 2-DoF manipulator (see Section \ref{section: Exm1} for more details) based on Fig. \ref{fig:F_Generator_Fault} (b). 
We used the RISE controller \cite{Xian2004} as a robust controller in the structure of Fig. \ref{fig:F_Generator_Fault} (b). We considered 50 percent parametric uncertainty for the parameters' value of the 2-DoF manipulator in the training data set for ideal conditions.
According to Fig. \ref{fig:F_CGAN} (a), we can suppose the interval $[0.45,0.55]$ that is symmetric interval around $\frac{1}{2}$ as an ideal region for the manipulator. 

\begin{figure}[!h]
  \centering
  \includegraphics[width=1\columnwidth]{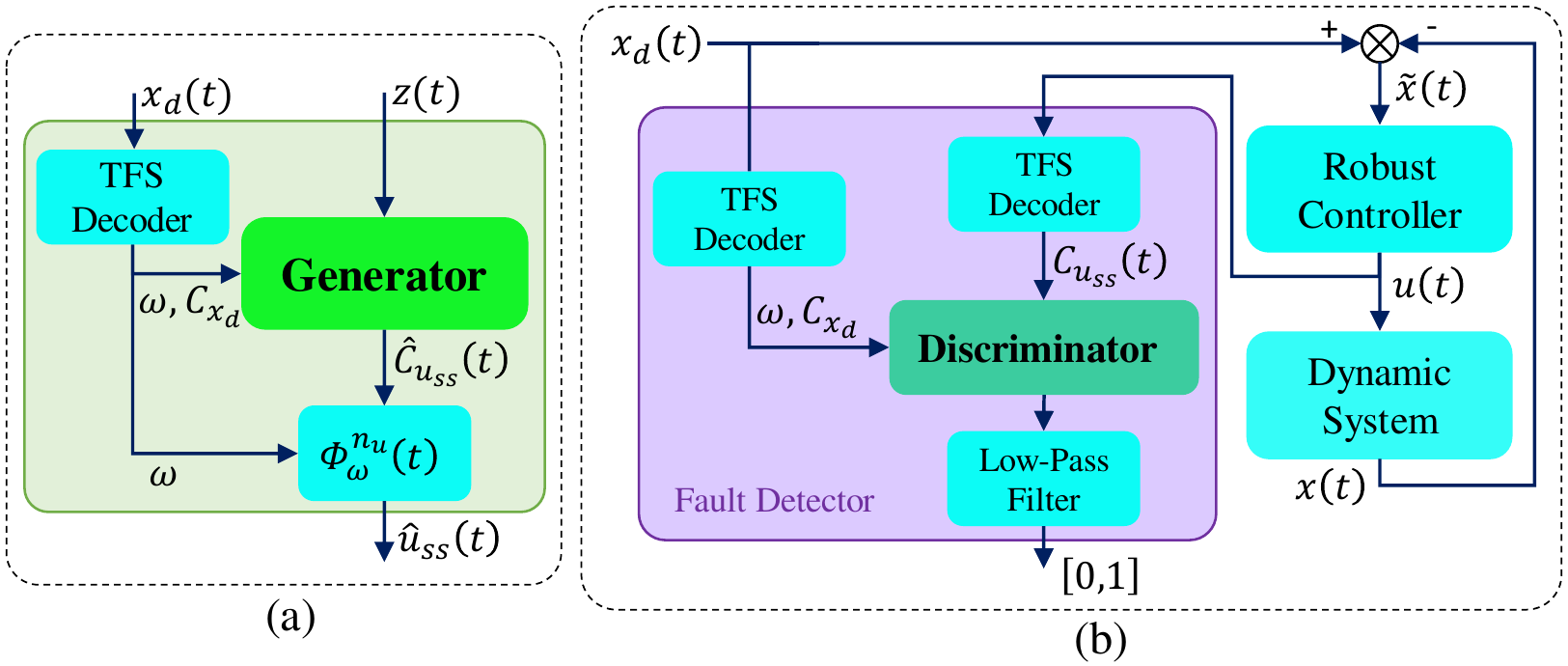}
  \caption{CGAN applications after training: (a) using the conditional generator as a dynamic policy for steady-state control, (b) using the conditional discriminator as a fault detector.}\label{fig:F_Generator_Fault}
\end{figure}

\begin{figure}[!h]
  \centering
  \includegraphics[width=1\columnwidth]{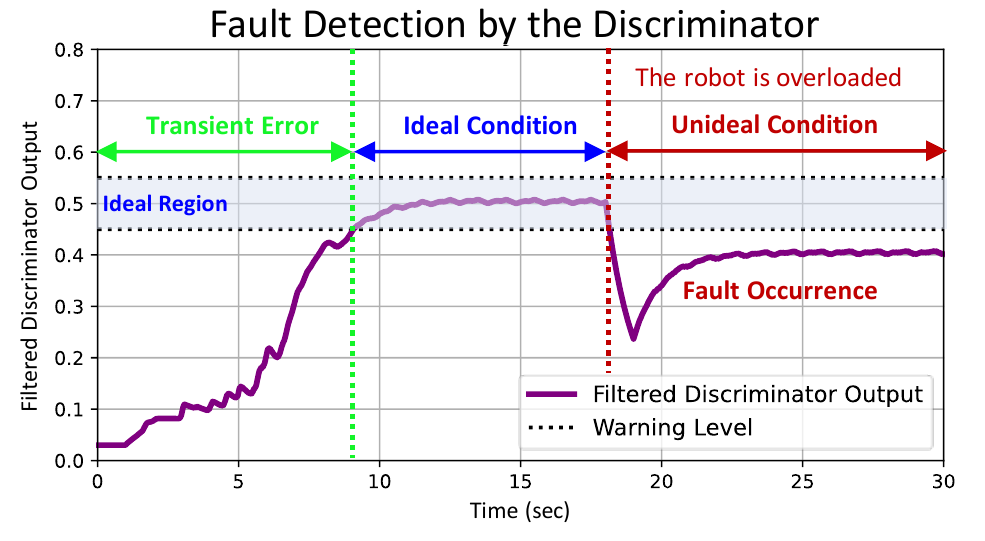}
  \caption{Filtered discriminator output for a 2-DoF manipulator in the closed-loop control. Initially, there is a transient state error until the system reaches its ideal condition. After the 18th second, the robot is overloaded by 60 percent of the second link's mass and moves away from its ideal condition. Therefore, a fault occurs, and the discriminator detects it.}\label{fig:F_Fault_Curve}
\end{figure}

%%%%%%%%%%%%%%%%%%%%%%%%%%%%%%%%%%%%%%%%%%%%%%%%%%%%%%%%%%%%%%%%%%%%%%%%%%%%%%%%%%%%%%%%%%%%%%%%%%%%%%

\subsection{RISE: Regulation Control Policy}

The RISE controller
\cite{Xian2004} is a robust control that yields asymptotic tracking stability for tracking tasks. RISE uses the integral of the sign of error term to compensate for the uncertain system dynamics.
Based on the work of \cite{Yang2015} and \cite{Xian2004}, the RISE control law is as follows:
\begin{align}\label{eq17}
u_{reg} \triangleq - (k+1) e(t) - \int^t_0 \left[\alpha (k+1) e(\tau)+\beta sgn(e(\tau))\right] d\tau,
\end{align}
where $\alpha$, $\beta$, and $k$ are positive control gains that must be chosen according to Section \ref{section: Stability Analysis} to make the closed-loop system stable and robust, $sgn(.)$ expresses the standard sign function, and $e(t)$ is a filtered tracking error as follows:
\begin{align}\label{eq18}
e(t) &\triangleq \lambda_{m-1} \tilde{x}^{(m-1)}(t) +\lambda_{m-2} \tilde{x}^{(m-2)}(t)+...+\lambda_{0}\tilde{x}(t)\nonumber \\
&= \sum^{m-1}_{i=0} \lambda_{i} \tilde{x}^{(i)}(t),
\end{align}
where $\lambda_{0}, ..., \lambda_{m-1}$ are suitable Hurwitz constants. Note that $e(t)$ depends only on the extended state vector $X(t)$ in \eqref{eq3.9}; hence, it is measurable. 

%%%%%%%%%%%%%%%%%%%%%%%%%%%%%%%%%%%%%%%%%%%%%%%%%%%%%%%%%%%%%%%%%%%%%%%%%%%%%%%%%%%%%%%%%%%%%%%%%%%%%%

\subsection{Adaptation Mechanisms}
As mentioned previously, by adjusting the low-dimensional input noise of the conditional generator, we can achieve the appropriate patterns for the dynamic system. In comparison, conventional adaptive control methods explore uncertainties in the original high-dimensional space \cite{Slotine1991, chen2015, Astrom1994}.

Two approaches can be used to tune the input noise: indirect and direct. In the indirect approach, we first estimate uncertainties by a system identifier that knows the structure of the dynamic model. Then, the estimated uncertainties map to the input noise. We consider the conditional generator as a dynamic policy in the Meta-RL framework in the direct approach, where an RL strategy will adjust the input noise without any knowledge of the system.

%%%%%%%%%%%%%%%%%%%%%%%%%%%%%%%%%%%%%%%%%%%%%%%%%%%%%%%%%%%%%%%%%%%%%%%%%%%%%%%%%%%%%%%%%%%%%%%%%%%%%%

\subsubsection{\textbf {Direct Approach} (Novel Meta-RL Structure)}

In the RL framework, a deterministic policy function gets the state of the dynamic system (environment) and returns control action for manipulation. An immediate reward function concerning the state and action gives a real scalar number as the immediate performance of the policy. Value functions are introduced to evaluate the degree to which a policy is optimal in a given state or a given state--action pair. The goal is to find a policy function that optimizes the value functions \cite{Sutton2017}. Meanwhile, in the Meta-RL framework, the policy is dynamic, and there is a meta-training loop to adjust the parameters of the dynamic policy according to variations in the dynamic system \cite{Finn2017, Rakelly2019, Nagabandi2019}. In other words, the goal is to find an adaptation rule that tunes the dynamic policy to optimize the value functions.

This paper introduces a novel Meta-RL structure that leverages the conditional generator as a dynamic parameterized policy. The conditional generator provides a set of policies that have already been learned for different situations of the uncertain system. Hence, it can reduce the search space to a low-dimensional latent space compared to the traditional policy gradient methods that search the whole parameter space of a general parameterized policy. Therefore, we use a feedforward NN as an adaptor network that updates the input noise according to the interactions with the system (Fig. \ref{fig:F_Dir}). Fig. \ref{fig:F_Generator_Fault} (a) shows the structure employed as the actor, and a feedforward NN is used as the critic to estimate the state--action value function. 

We define $X(t)$ in \eqref{eq3.9} as the state vector, $u(t)$ as the action vector, and the immediate reward signal $r(t)$ according to cost function $J(.)$ for the RL framework, as follows:
\begin{align}
r(t) &\triangleq \exp\left(-J\left(X(t),u(t)\right)\right)\in \mathbb{R}^+,\label{eq3.19} \\ 
J(X(t),u(t))&\triangleq Q||\tilde{x}(t)||_1+R||u(t)||_1 \in \mathbb{R}^+,\label{eq3.192}
\end{align}
where $Q$ and $R$ are positive diagonal matrices that sufficiently normalize control effort and tracking error terms, and $||.||_1$ stands for the $L1$ norm. The state--action value function $V(.)$ is defined as 
\begin{align}\label{eq3.199}
V(X(t),u(t)) = \int_t^\infty e^{-\gamma (\tau-t)} r(\tau) d\tau,
\end{align}
where $\gamma$ is a positive attenuation rate. We use the continuous temporal difference (TD) method according to \cite{LEE2021} to update the weights of the critic network as follows:
\begin{align}\label{eq3.20}
\delta_{TD}(t) &\triangleq r(t) + \dot{\hat{V}}_v(X(t),u(t)) -\gamma \hat{V}_v(X(t),u(t))\nonumber\\
\dot{v}(t) &= \eta_v \delta_{TD}(t) \frac{\partial \hat{V}_v(X(t),u(t))}{\partial v},
\end{align}
where $\delta_{TD}$ is temporal difference error, $\hat{V}_v(.)$ is the estimate of the state--action value function with weights $v$, and $\eta_v>0$ is the learning rate of the critic network. The adaptor network $\psi_w(.)$ with weights $w$ transforms the system state vector to update the input noise of the conditional generator such that 
\begin{align}\label{eq3.21}
\dot{z}(t) &= \psi_w (X(t)),
\end{align}
where $z(t)$ is a functional of the $w$ and $X$.
Based on \eqref{eq3.11}, \eqref{eq16}, and \eqref{eq3.21}, we obtain the gradient of the policy w.r.t. $w$ by using the chain rule as follows:
\begin{align}\label{eq3.22}
\frac{\partial u}{\partial w} &= \frac{\partial \hat{u}_{ss}}{\partial w}=\frac{\partial \hat{u}_{ss}}{\partial \hat{C}_{u_{ss}}} . \frac{\partial \hat{C}_{u_{ss}}}{\partial z} . \frac{\partial z}{\partial w} \nonumber \\
&= \bar\Phi^{n_u}_\omega.\frac{\partial Gen(z|C_{x_d},\omega)}{\partial z}. \int_0^t \frac{\partial \psi_w (X)}{\partial w} d\tau,
\end{align}
where $\bar\Phi^{n_u}_\omega$ is a tensor related to the elements of the vector $\Phi^{n_u}_\omega$; $\partial Gen(z|C_{x_d},\omega)/\partial z$ indicates the Jacobian of the conditional generator w.r.t. $z$; and $\partial \psi_w (X)/\partial w$ is the gradient of the adaptor network w.r.t. $w$. Therefore, the weights of the adaptor network can be updated using the DDPG method \cite{Silver2014 ,Lillicrap2016a}, and \eqref{eq3.22} as follows:
\begin{align}\label{eq3.23}
\dot{w}(t) &= \eta_w  \frac{\partial \hat{V}_v(X,u)}{\partial u(t)}. \frac{\partial u}{\partial w}
\nonumber\\ &=\eta_w  \frac{\partial \hat{V}_v(X,u)}{\partial u}.\bar\Phi^{n_u}_\omega.\frac{\partial Gen(z|C_{x_d},\omega)}{\partial z}.\int_0^t \frac{\partial \psi_w (X)}{\partial w} d\tau,
\end{align}
where $\partial \hat{V}_v(X,u)/\partial u$ indicates the Jacobian of the critic, and $\eta_w>0$ is the learning rate of the adaptor network. While the adaptor network is updating, we fix the conditional generator.  

We implement the update rules \eqref{eq3.20} and \eqref{eq3.23} in an online manner (Fig. \ref{fig:F_Dir}). They can also be executed by using mini-batch data prioritized by time and collected in a reply buffer (RB), as shown in Fig. \ref{fig:F_Dir}. Experience replay buffer technique improves the data efficiency and avoids local minimum \cite{Chowdhary2010, Adam2012,ZHANG202040}. 
\begin{figure}[!h]
  \centering
  \includegraphics[width=0.9\columnwidth]{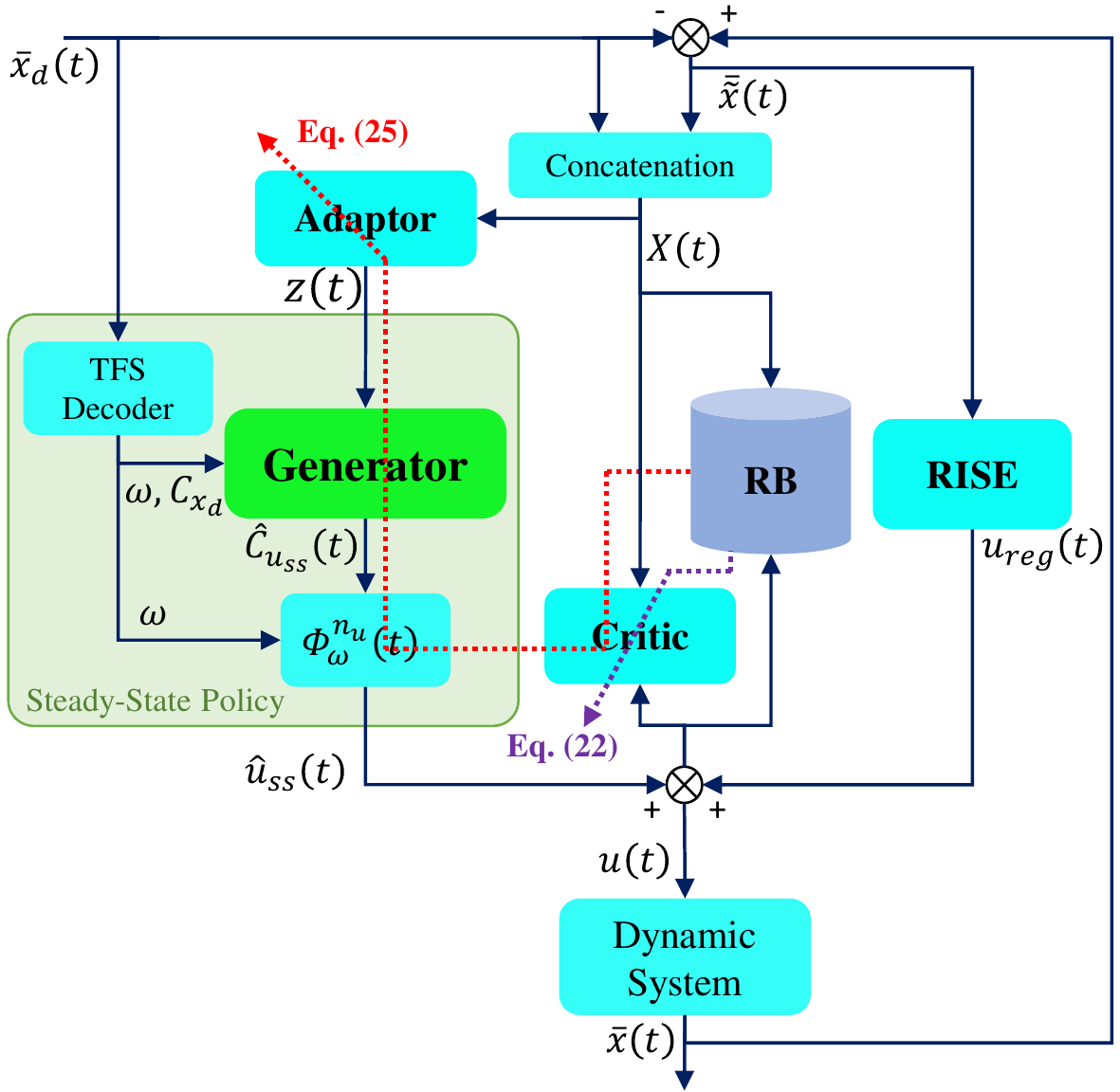}
  \caption{The direct adaptation mechanism for fast and optimal adaptive tracking control scheme: RISE controller and the conditional generator are used for regulation and the steady-state policy, respectively. The novel Meta-RL structure adjusts the steady-state policy.}\label{fig:F_Dir}
\end{figure}

We exploit the RL framework to approximately optimally update the weights of the adaptor network such that the value function in \eqref{eq3.199} is maximized. Therefore, we introduced a Meta-RL structure incorporated into the RISE controller that yields a fast and optimal adaptive control scheme. As mentioned previously, we can use the structure in Fig. \ref{fig:F_Generator_Fault} (a) as a data-driven adaptable CPG. Therefore, we can employ the structure in Fig. \ref{fig:F_Dir} to adjust a versatile CPG-based control.

%%%%%%%%%%%%%%%%%%%%%%%%%%%%%%%%%%%%%%%%%%%%%%%%%%%%%%%%%%%%%%%%%%%%%%%%%%%%%%%%%%%%%%%%%%%%%%%%%%%%%%
\subsubsection{\textbf {Indirect Approach} (EKF as a System Identifier)}

EKF is one of the most applicable model-based identifiers for uncertain nonlinear systems.
We use EKF as a system identifier to simultaneously estimate parametric uncertainties and filtered state vector of the dynamic system based on \cite{Boutayeb1997}. The formulation and settings described by \cite{Boutayeb1997} are employed for better convergence.

This subsection presents a novel indirect adaptive control using EKF and the conditional generator. 
EKF can estimate uncertainties as a feature vector of the uncertain dynamic system. Therefore, to adapt the steady-state policy to the uncertain system, an interface is required to decode the uncertainty estimated by EKF to the latent space prepared by the input noise of the conditional generator.

We employ a feedforward NN with weights $w$ as an adaptor ($\psi_w(.)$) to decode the high-dimensional estimated uncertainties into the low-dimensional input noise such that 
\begin{align}\label{eq3.244}
z(t) &= \psi_w (\hat{\theta}(t)),
\end{align}
where $\hat{\theta}(t)$ is the estimated parametric uncertainty vector by EKF at time $t$. 
we apply supervised learning to update the weights of the adaptor. To do so, we transform the data set $\mathscr{D}_1$ into the following data set $\mathscr{D}_3$ by using Fourier series transformation and EKF: 
\begin{align}\label{eq3.25}
\mathscr{D}_3 = \{label: (\hat{\theta}^i, C^i_{x_d}, \omega^i),\:   target: C^i_{u_{ss}}| i=1,...,N\},
\end{align}
where $\hat{\theta}^i$ is the estimated uncertainty vector according to $i^{\text{th}}$ sample. Then, we consider the following loss function:
\begin{align}\label{eq3.26}
Loss(u_{ss},\hat{u}_{ss})&=\frac{1}{T} \int_T\| u_{ss}(t) - \hat{u}_{ss}(t) \|^2  dt,
\end{align}
By substituting \eqref{eq3.13} and \eqref{eq16} into \eqref{eq3.26} and with manipulations, we obtain
\begin{align}\label{eq3.27}
Loss(u_{ss},\hat{u}_{ss})&=trace((C_{u_{ss}}-\hat{C}_{u_{ss}}) P  (C^{T}_{u_{ss}}-\hat{C}^{T}_{u_{ss}}))\nonumber \\ &+||\delta_2(u_{ss})||_2^2.
\end{align}
where the first term is a positive quadratic function, the second term denotes the fixed TFS truncation error in \eqref{eq3.13}, and $P$ is the diagonal positive definite matrix 
\begin{align}
P = diag([1,0.5,0.5,...,0.5]) \in\mathbb{R}^{(2n_u+1) \times (2n_u+1)}.
\end{align}
By using \eqref{eq16}, \eqref{eq3.244}, and the chain rule, the gradient of the loss function w.r.t. the weights of the adaptor network is obtained according to Fig. \ref{fig:F_Adp} as follows:
\begin{align}\label{eq3.28}
\frac{\partial Loss(u_{ss},\hat{u}_{ss})}{\partial w}=\frac{\partial Loss(u_{ss},\hat{u}_{ss})}{\partial \hat{C}_{u_{ss}}} . \frac{\partial \hat{C}_{u_{ss}}}{\partial z} . \frac{\partial z}{\partial w}.
\end{align}
We implement the Adam optimization method \cite{kingma2017adam} to train the adaptor network using \eqref{eq3.28} while the conditional generator is fixed. 
\begin{figure}[!h]
  \centering
  \includegraphics[width=0.50\columnwidth]{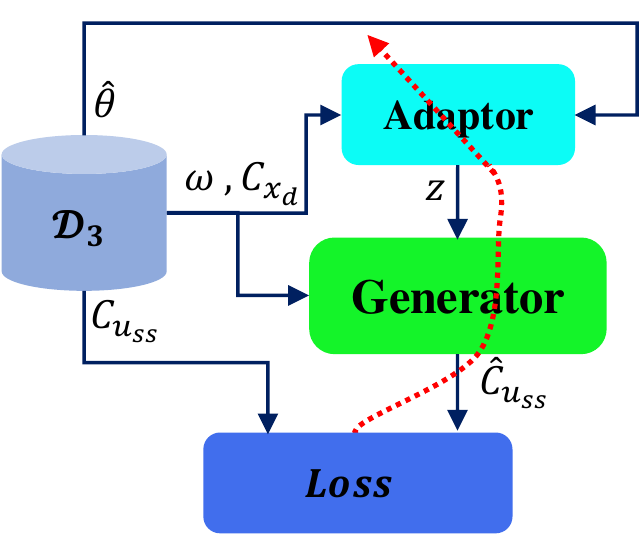}
  \caption{Supervised learning for the adaptor in the indirect adaptation mechanism.}\label{fig:F_Adp}
\end{figure}
Finally, Fig. \ref{fig:F_Ind} shows an indirect adaptive control after the adaptor and conditional generator network training. 
The known dynamic model structure requirement for EKF is a limitation of this approach. 
\begin{figure}[!h]
  \centering
  \includegraphics[width=0.81\columnwidth]{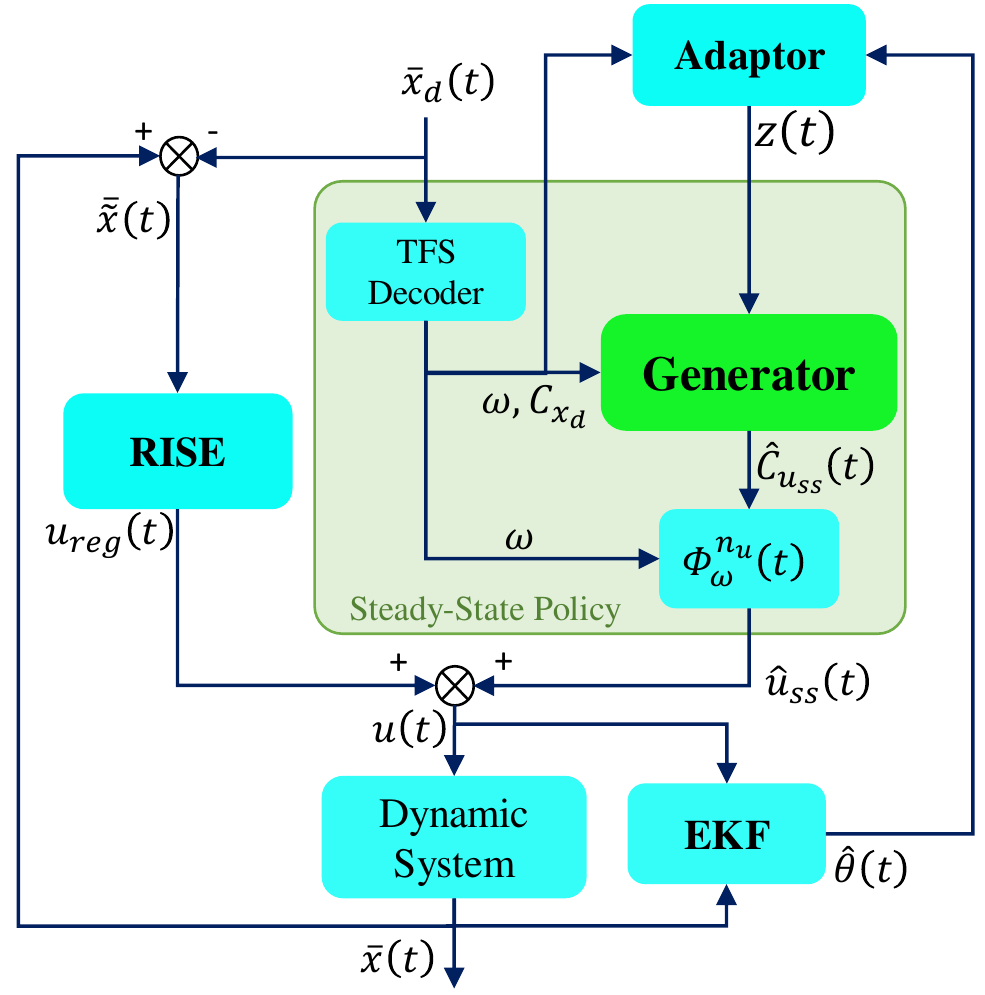}
 \caption{The indirect adaptation mechanism for indirect adaptive tracking control scheme: RISE controller and the conditional generator are used for regulation and the steady-state policy, respectively. The learned adaptor and EKF adjust the steady-state policy.}\label{fig:F_Ind}
\end{figure}
%%%%%%%%%%%%%%%%%%%%%%%%%%%%%%%%%%%%%%%%%%%%%%%%%%%%%%%%%%%%%%%%%%%%%%%%%%%%%%%%%%%%%%%%%%%%%%%%%%%%%%

\section{Stability Analysis}  \label{section: Stability Analysis}
In this section, we present a novel stability analysis for the direct adaptive approach as a fast and optimal adaptive tracking control scheme by exploiting the stability analysis approaches described by \cite{Yang2015,Xian2004,Chowdhary2010, SOKOLOV2015, ZHANG202040}.
Equations \eqref{eq2.1}, \eqref{eq3.11}, \eqref{eq16}, \eqref{eq17}, \eqref{eq18}, \eqref{eq3.20}, \eqref{eq3.21}, and \eqref{eq3.23} are the governing equations in this control system.

%%%%%%%%%%%%%%%%%%%%%%%%%%%%%%%%%%%%%%%%%%%%%%%%%%%

\subsection{Closed-Loop Extended Error Dynamics}
By recalling Assumption \ref{Assumption_G}, we can rewrite the dynamic system \eqref{eq2.1} as 
\begin{align}\label{eq4.1}
A(\bar{x})x^{(m)}(t)-B(\bar{x})-D(t)=u(t)
\end{align}
where $A(\bar{x})\triangleq g^{-1}(\bar{x})$, $B(\bar{x}) \triangleq g^{-1}(\bar{x})f(\bar{x})$, and $D(t) \triangleq g^{-1}(\bar{x}) d(t)$. Define error manifold $s(t)$ as follows:
\begin{align}\label{eq4.2}
s(t) \triangleq \dot{e}(t)+\alpha e(t),
\end{align}
where $e(t)$ and $\alpha$ are introduced in \eqref{eq18}. By differentiating \eqref{eq4.1} and \eqref{eq4.2} w.r.t. time and using \eqref{eq18}, we obtain 
\begin{align}\label{eq4.3}
A(\bar{x})\dot{s}(t) &= A(\bar{x}) \left( -x_d^{(m+1)}(t) +\lambda_{m-2} \tilde{x}^{(m)}(t)\right. \nonumber \\
&\left. +\lambda_{m-3} \tilde{x}^{(m-1)}(t)+...+\lambda_{0}\tilde{x}^{(2)}(t)+\alpha e(t) \right)  \nonumber \\
&- \dot{A}(\bar{x})x^{(m)}(t) +\dot{B}(\bar{x}) + \dot{D}(t)+\dot{u}(t),
\end{align}
where $\lambda_{m-1}=1$ without loss of generality. 

\begin{assumption} \label{Assumption_ErrorUss}
The estimation error between the estimated steady-state control in \eqref{eq16} and the exact value of the steady-state control in \eqref{eq3.8} $(\delta_3(t)\triangleq \hat{u}_{ss}(t)-u_{ss} (t))$, and its derivatives are norm-bounded as
\begin{align}\label{eq4.3333}
\|\delta_3(t)\|_1<\bar{\delta}_3, \quad \|\dot{\delta}_3(t)\|_1<\dot{\bar{\delta}}_3, \quad \|\ddot{\delta}_3(t)\|_1<\ddot{\bar{\delta}}_3,
\end{align}
where $\bar{\delta}_3$, $\dot{\bar{\delta}}_3$, and $\ddot{\bar{\delta}}_3$ are positive constants. $\delta_3(t)$ can originate by the estimation error of the conditional generator, adaptor error, and the TFS truncation error. 
\end{assumption}

By differentiating $u_{reg}$ in \eqref{eq17}, and based on \eqref{eq3.11}, \eqref{eq4.2}, and Assumption \ref{Assumption_ErrorUss}, we can write the derivative of the control input as follows:
\begin{align}\label{eq4.343}
\dot{u} = \dot{u}_{ss}(t) + \dot{\delta}_3(t) - (k+1) s(t)-\beta sgn(e(t)),
\end{align}
Let us consider the auxiliary function $N(\bar{x},\bar{x}_d)$ as follows:
\begin{align}\label{eq4.4}
N(\bar{x},\bar{x}_d) &\triangleq -A(\bar{x}) \left( -x_d^{(m+1)}(t) +\lambda_{m-2} \tilde{x}^{(m)}(t)\right. \nonumber \\
&\left. +\lambda_{m-3} \tilde{x}^{(m-1)}(t)+...+\lambda_{0}\tilde{x}^{(2)}(t)+\alpha e(t) \right)  \nonumber \\
&+ \dot{A}(\bar{x})x^{(m)}(t) -\dot{B}(\bar{x}) - \frac{1}{2} \dot{A}(\bar{x})s(t) - e(t),
\end{align}
so that $N_d(t) \triangleq N(\bar{x}_d,\bar{x}_d)$ is equal to the derivative of the steady-state control in \eqref{eq3.8}, given by 
\begin{align}\label{eq4.5}
N_d(t) = A(\bar{x}_d)x_d^{(m+1)}(t) + \dot{A}(\bar{x}_d)x_d^{(m)}(t) -\dot{B}(\bar{x}_d) = \dot{u}_{ss}(t),
\end{align}
After substituting \eqref{eq4.343} and \eqref{eq4.4} into \eqref{eq4.3} and adding and subtracting $N_d(t)$ to the right-hand side of  the resulting equation, we obtain the following closed-loop extended tracking error dynamics:
\begin{align}\label{eq4.6}
A(\bar{x})\dot{s}(t) &= \tilde{N}(\bar{x},\bar{x}_d) - \frac{1}{2} \dot{A}(\bar{x})s(t) - e(t) \nonumber \\
&- (k+1) s(t)-\beta sgn(e(t)) + \dot{D}(t)+\dot{\delta}_3(t),
\end{align}
where $\tilde{N} \triangleq N_d-N$.

\begin{remark}
Because $N(\bar{x},\bar{x}_d)$ is continuously differentiable, in Appendix of \cite{Queiroz1997}, the following inequality is proved by the mean value theorem:
\begin{align}\label{eq4.7}
\| \tilde{N}\| \leq \rho (\|\xi\|) \|\xi\|,
\end{align}
where $\xi(t)\triangleq \left[s^T(t), \: e^T(t) \right]^T\in\mathbb{R}^{2n} $ is the extended tracking error vector, and $\rho (.)$ represents a non-negative, non-decreasing invertible function.
\end{remark}

%%%%%%%%%%%%%%%%%%%%%%%%%%%%%%%%%%%%%%%%%%%%%%%%%%%

\subsection{Critic and Adaptor Structure}
To simplify the convergence analysis, we consider the following structure for the critic and adaptor network, respectively. 
\begin{align}
\hat{V}_v\left(X(t),u(t)\right) &\triangleq v^T \sigma \left(\Upsilon(t) \right), \label{eq4.81}\\
\dot{z}_j (t) = \psi_{w_j}(X(t)) &\triangleq w_j^T \dot{e}(t) \label{eq4.82},
\end{align}
where $\Upsilon(t)=\Gamma \left[X^T(t),u^T(t)\right]^T$; $v\in \mathbb{R}^{l_c}$ and $\Gamma \in \mathbb{R}^{l_c \times (2nm+n)}$ are the weights of the output layer and hidden layer for the critic network, respectively.
We choose hidden layer weight $\Gamma$ randomly by the Glorot initialization method \cite{pmlr-v9-glorot10a} and keep it unchanged.
$z_j$ and $\psi_{w_j}$ for $j=1,...,l_z$ denote the $j^{\text{th}}$ element of the vector $z$ and $\psi_{w}$, respectively; and $l_z$ signifies the dimension of the latent variable $z$. 
$w_j \in \mathbb{R}^{n}$ for $j=1,...,l_z$ are the weights of the adaptor network.
Let $\sigma(.)$ be the standard $\tanh(.)$ function that indicates the activation function of the critic network.

We can rewrite the learning rule \eqref{eq3.20} by using \eqref{eq4.81} and recorded mini-batch data as follows:
\begin{align}\label{eq4.10}
\dot{v}(t) &= \eta_v a(t) \left(r(t)+ (\dot{a}^T(t)- \gamma a^T(t))v(t) \right)\nonumber \\ 
&+\sum^p_{i=0} \eta_v a(t_i) \left(r(t_i)+ (\dot{a}^T(t_i)- \gamma a^T(t_i))v(t_i) \right),
\end{align}
where $a(t)\triangleq \sigma(\Upsilon(t))$, $\dot{a}(t)\triangleq \dot{\sigma}(\Upsilon(t))$, $\|a(t)\|\le 1$, and $\|\dot{a}(t)\| \le \dot{\bar{a}}$. Also, we consider the following verifiable condition on online recorded data similar to \cite{Chowdhary2010} to ensure convergence.
\begin{align}\label{eq4.11}
&rank\left([a(t_1),a(t_2),... ,a(t_p)]\right)=l_c.
\end{align}
\begin{remark} \label{Remark3}
Based on \eqref{eq4.11}, matrix $\Pi \triangleq \sum_{i=0}^p a(t_i)a^T(t_i)$ is positive definite.
\end{remark}
Further, we can rewrite the update rule \eqref{eq3.23} based on \eqref{eq4.82} as
\begin{align}\label{eq4.12}
\dot{w}_j(t)&= v^T(t) \Lambda_j(t) e(t),
\end{align}
where $\Lambda_j \in \mathbb{R}^{l_c}$ is given by
\begin{align}\label{eq4.13}
\Lambda_j \triangleq \eta_w \sigma'(\Upsilon) \Gamma 
\left[
\begin{array}{cc}
0_{2nm \times n_u}\\ \hline
\partial Gen(z|C_{x_d},\omega)/\partial z_j
\end{array}\right]
\Phi^{n_u}_\omega.
\end{align}
Therefore, the implementable version of \eqref{eq4.82} can be obtained by the integration by parts as follows:
\begin{align}\label{eq4.14}
z_j (t) = w_j^T(t)e(t)-w_j^T(0)e(0) - \int_0^t v^T(\tau) \Lambda_j(\tau) \|e(\tau)\|_2^2 d\tau,
\end{align}

\begin{assumption}
Assume that $z^*$, $v^*$ and $w^*$ are the bounded optimal values for $z$, $v$ and $w$, respectively. Define the estimation errors as $\tilde{z}\triangleq z-z^*$, $\tilde{v}\triangleq v-v^*$ and $\tilde{w}\triangleq w-w^*$. 
Therefore, the optimal value function and the optimal adaptor network is as follows:
\begin{align}\label{eq4.15}
V^*_v\left(X(t),u(t)\right) &= v^{*^T} \sigma \left(\Upsilon(t) \right) +\delta_4(t) \nonumber \\
\psi^*_{w_j}(X(t)) &= w_j^{*^T} \dot{e}(t) +\delta_5(t),
\end{align}
where $\delta_4(t)$ and $\delta_5(t)$ are the reconstruction error for the critic and the adaptor network, respectively. Assume $\delta_4(t)$ and $\delta_5(t)$ and their derivatives are norm-bounded as 
\begin{align}\label{eq4.16}
\|\delta_4(t)\| &\le \bar{\delta}_4, \quad \|\dot{\delta}_4(t)\| \le \dot{\bar{\delta}}_4, \\
 \|\delta_5(t)\| &\le \bar{\delta}_5, \quad \|\dot{\delta}_5(t)\| \le \dot{\bar{\delta}}_5,
\end{align}
where $\bar{\delta}_4$, $\bar{\delta}_5$, $\dot{\bar{\delta}}_4$, and $\dot{\bar{\delta}}_5$ are positive constants. 
\end{assumption}
Based on \cite{Sutton2017} and \cite{LEE2021}, the TD error \eqref{eq3.20} for the optimal value function equals zero. Hence, we obtain 
\begin{align}\label{eq4.17}
r(t) + v^{*^T} \dot{\sigma} \left(\Upsilon(t) \right) - \gamma v^{*^T} \sigma \left(\Upsilon(t) \right) = \gamma \delta_4(t) - \dot{\delta}_4(t).
\end{align}
%%%%%%%%%%%%%%%%%%%%%%%%%%%%%%%%%%%%%%%%%%%%%%%%%%%%%%%%%%%%%%%%%%%%%%%%%%%%%%%%%%

\subsection{Lyapunov Approach for Stability Analysis}
We first state the following lemma similar to \cite{Yang2015,Xian2004}. Then, the stability analysis of the closed-loop system will be presented.

\begin{lemma} \label{lemma1}
Consider an extra state variable $\Omega(t)$ to augment the state vector as follows
\begin{align}\label{eq4.18}
\dot{\Omega}&\triangleq - s^T \left(\dot{\delta}_3+\dot{D}-\beta_d sgn(e)\right)-M^T_1\dot{e}(t)-M^T_2e,
\end{align}
where $\beta_d$ is a positive constant; $M_1\triangleq\sum_{j=1}^{l_z}\tilde{z}_j w_j$ and $M_2\triangleq\sum_{j=1}^{l_z} \Lambda^T_j v \tilde{w}_j$ such that $\|M_1\|_1<c_{M_1}$, $\|\dot{M}_1\|_1<c_{\dot{M}_1}$, and $\|M_2\|_1<c_{M_2}$. 
If the following equations are satisfied, then $\Omega(t) \ge 0$.  
\begin{align}
\Omega(0)&= \beta_d \|e(0)\|_1-e^T(0)(\dot{\delta}_3(0)+\dot{D}(0)+M_1(0)),  \label{eq4.19}\\
\beta_d &\ge \dot{\bar{\delta}}_3 +c_{d_1} +c_{M_1} + \frac{1}{\alpha} \left(\ddot{\bar{\delta}}_3+c_{d_2}+c_{\dot{M}_1}+c_{M_2}\right). \label{eq4.20}
\end{align}
\end{lemma}
\begin{proof}
Please refer to Supplementary Material, Section I. 
\end{proof}
\begin{theorem} \label{Theorem1}
By choosing positive-definite Lyapunov function
\begin{align}\label{eq4.22}
L(\mu,t) \triangleq &\Omega+\frac{1}{2}(e^Te+s^TA(\bar{x})s+\frac{2\gamma}{\eta_v}\tilde{v}^T\tilde{v}+\sum_{j=0}^{l_z}\tilde{w}^T_j \tilde{w}_j  \nonumber\\
& +\tilde{z}^Tz), \nonumber\\
&\alpha_{min}\|\mu\|^2_2 \le L(\mu,t) \le \alpha_{max}\|\mu\|^2_2,
\end{align}
where $\mu \triangleq \left[\xi^T,\, \tilde{v}^T,\,\tilde{z}^T,\,\tilde{w}^T_1,\,...,\tilde{w}^T_{l_z},\,\sqrt{\Omega}\right]^T$, $\alpha_{min} \triangleq  \frac{1}{2}\min\left\{1,\,1/\eta_v,\,1/\bar{g}\right\}$ , and $\alpha_{max} \triangleq \max\left\{1,\,1/\eta_v,\,1/2\underline{g}\right\}$,
the controller and the update laws given in \eqref{eq3.11}, \eqref{eq16}, \eqref{eq17}, \eqref{eq18}, \eqref{eq4.10}, \eqref{eq4.12}, and \eqref{eq4.14} ensure that all system signals
are bounded if the following inequalities are satisfied. Also, the tracking error can converge to zero asymptotically if $\beta=\beta_d$. 
\begin{gather}
\frac{|\beta-\beta_d|}{\varepsilon} \le \|\xi\| \le \rho^{-1}\left(2\sqrt{k(\bar{\alpha}-\varepsilon)}\right)\label{eq4.26},\\
0 < \varepsilon < \bar{\alpha} \le 1, \label{eq4.27}\\
\|\tilde{v}^Ta\| \ge \frac{\gamma \bar{\delta}_4 + \dot{\bar{\delta}}_4}{\gamma},\label{eq4.28}\\
\|\tilde{v}\| \ge \frac{2 \gamma p(\gamma \bar{\delta}_4 + \dot{\bar{\delta}}_4)}{2\gamma^2 \lambda_{\min}(\Pi) -2 \gamma p \dot{\bar{a}}-\dot{\bar{a}}},\label{eq4.29}\\
2\gamma^2 \lambda_{\min}(\Pi) -2 \gamma p \dot{\bar{a}}-\dot{\bar{a}} \ge 0,  \label{eq4.30}
\end{gather}
where $\bar{\alpha}\triangleq\min\{\alpha,1\}$ and $\lambda_{\min}(\Pi)$ denotes the minimum eigenvalue of matrix $\Pi$.
\end{theorem}
\begin{proof}
Please refer to Supplementary Material, Section II. 
\end{proof}

According to \eqref{eq4.22} and \eqref{eq4.26} and based on Theorem 4.18 of \cite{Khalil2002}, an approximation for the RoA is given by 
\begin{align}\label{eq4.31}
\|\xi(t_0)\| \le \sqrt{\frac{\alpha_{\min}}{\alpha_{\max}}}\rho^{-1}\left(2\sqrt{k(\bar{\alpha}-\varepsilon)}\right);
\end{align}
also, we can obtain ultimate bounds (UB) on the extended tracking error $\xi$ (if $\beta \ne \beta_d$) and the value estimation error $\tilde{v}$ as follows:
\begin{gather}
\textrm{UB}_{\xi} = \sqrt{\frac{\alpha_{\max}}{\alpha_{\min}}} \left(\frac{|\beta-\beta_d|}{\varepsilon}\right) \label{eq4.32}, \\
\textrm{UB}_{\tilde{v}} = \sqrt{\frac{\alpha_{\max}}{\alpha_{\min}}} \max\left\{\frac{\gamma \bar{\delta}_4 + \dot{\bar{\delta}}_4}{\gamma} ,\frac{2 \gamma p(\gamma \bar{\delta}_4 + \dot{\bar{\delta}}_4)}{2\gamma^2 \lambda_{\min}(\Pi) -2 \gamma p \dot{\bar{a}}-\dot{\bar{a}}}\right\}.\label{eq4.33}
\end{gather}
Based on the presented analysis in the previous subsection, some theoretic remarks for the proposed controllers are as follows:
\begin{itemize}

\item Based on Theorem \ref{Theorem1} and its results, the estimation errors of the critic and the adaptor networks are bounded. Therefore, they are approximately optimal ultimately.

\item The convergence of the presented Meta-RL framework does not affect the tracking error convergence. In other words, it is the regulation controller which preserves the stability of the closed-loop error dynamics and robustness to bounded disturbances. 

\item According to \eqref{eq4.31}, the RoA approximation depends on the control gains $k$ and $\alpha$. Therefore, the closed-loop system is semi-globally stable. 

\item According to Theorem 4.18 of \cite{Khalil2002} and \eqref{eq4.26}, we can determine the minimum required control gain $k$ for a specific initial error as follows:
\begin{align}\label{eq4.34}
k \ge \frac{1}{4(\bar{\alpha}-\varepsilon)}\rho^{2}\left(\max{\left\{\sqrt{\frac{\alpha_{\max}}{\alpha_{\min}}}\|\xi(t_0)\|,\: \textrm{UB}_{\xi} \right\}} \right).
\end{align}
Also, according to \eqref{eq4.20}, the minimum required control gain $\beta$ for asymptotic stability is as follows:
\begin{align}\label{eq4.35}
\beta \ge \dot{\bar{\delta}}_3 +c_{d_1} +c_{M_1} + \frac{1}{\alpha} \left(\ddot{\bar{\delta}}_3+c_{d_2}+c_{\dot{M}_1}+c_{M_2}\right).
\end{align}

\item Suppose the rates of change of the estimation error for the steady-state control ($\dot{\bar{\delta}}_3$ and $\ddot{\bar{\delta}}_3$) are significant. Then, according to \eqref{eq4.35} and \eqref{eq4.32}, the ultimate bound of the extended tracking error will be considerable, and the higher control gain $\beta$ will be required to make the closed-loop system asymptotically stable.

\item According to \eqref{eq4.27}, \eqref{eq4.32}, and \eqref{eq4.34}, if we consider a large value for $\varepsilon$ in the interval $(0,\bar{\alpha})$, then the ultimate bound of the extended tracking error will be small, and more considerable control gain $k$ will be required to make the closed-loop system stable.

\item Based on \eqref{eq4.31}, \eqref{eq4.32}, and \eqref{eq4.34}, the condition number of the matrix $g(\bar{x})$ can affect the system stability characteristics and the minimum required control gain $k$. 

\item For \eqref{eq4.30} to be held, $\lambda_{\min}(\Pi)$ should be large enough. Therefore, we need to collect rich data to increase $\lambda_{\min}(\Pi)$. In other words, we need more exploration and exciting data  to ensure that \eqref{eq4.30} holds. Also, when $\lambda_{\min}(\Pi)$ is large, based on \eqref {eq4.33}, the ultimate bound of the value function estimation error may be decreased. 

\item According to \eqref{eq4.33}, the ultimate bound of the value function estimation error directly relates to the reconstruction error of the critic network. 

\item We can similarly state this stability analysis for the presented indirect adaptive control. With the difference that we do not consider the last three terms in the Lyapunov function \eqref{eq4.22}.

\end{itemize}

%%%%%%%%%%%%%%%%%%%%%%%%%%%%%%%%%%%%%%%%%%%%%%%%%%%%%%%%%%%%%%%%%%%%%%%%%%%%%%%%%%%%%%%%%%%%%%%%%%%%%%

\section{Simulation Results} \label{section: Simulation Results}

This section consists of two parts, beginning with a 2-DoF robot manipulator simulation to verify and compare the proposed controllers' effectiveness with the previous work \cite{Yang2015}. Then, we demonstrate the scalability of the proposed control scheme by applying our novel Meta-RL structure to a large-scale musculoskeletal system.

%%%%%%%%%%%%%%%%%%%%%%%%%%%%%%%%%%%%%%%%%%%%%%%%%%%%%%%%%%%%%%%%%%%%%%%%%%%%%%%%%%%%%%%%%%%%%%%%%%%%%%
\subsection{Manipulator Robot} \label{section: Exm1}
This subsection presents the simulation results of the proposed control schemes on a 2-DoF manipulator shown in Fig. \ref{fig:F_2DOF}. To evaluate the effectiveness of the proposed methods, we compare the results with the previous method \cite{Yang2015}. The dynamics of the 2-DoF manipulator \cite{lewis2003robot} is given by 
\begin{align}\label{eq5.1111}
&(\eta_1+\eta_2+2\eta_3\cos(x_2))\ddot{x}_1 + (\eta_2+\eta_3\cos(x_2))\ddot{x}_2 \nonumber\\
&-\eta_3(2\dot{x}_1\dot{x}_2+\dot{x}_2^2)\sin(x_2)+\eta_4\eta_1\cos(x_1)\nonumber\\
&+\eta_3\eta_4\cos(x_1+x_2)+d_1+\tau_{f_1}=u_1 \nonumber\\
&(\eta_2+\eta_3\cos(x_2))\ddot{x}_1+\eta_2\ddot{x}_2+\eta_3\dot{x}^2_1\sin(x_2)\nonumber\\
&+\eta_3\eta_4\cos(x_1+x_2)+d_2+\tau_{f_2}=u_2,
\end{align}
where $\eta_1\triangleq(m_1+m_2)l_1^2$, $\eta_2\triangleq m_2l_2^2$, $\eta_3\triangleq  m_2l_1 l_2$, $\eta_4\triangleq g/l_1$, and $\tau_{f_i}\triangleq F_{s_i} sgn(\dot{x}_i) + F_{v_i} (\dot{x}_i)$. 
Tables \ref{table_1} and \ref{table_2} present the description and the parameters' value of the 2-DoF robot manipulator, respectively.

\begin{figure}[!h]
  \centering
  \includegraphics[width=0.6\columnwidth]{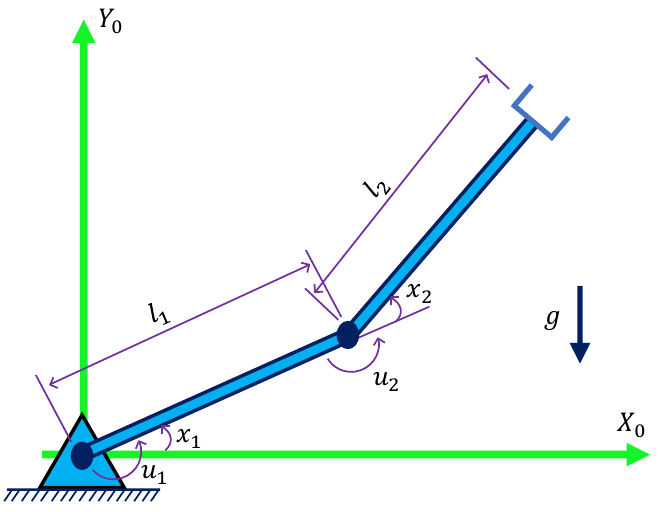}
  \caption{Diagram of the 2-DoF robotic manipulator.}\label{fig:F_2DOF}
\end{figure}

\begin{table}[!t]
\renewcommand{\arraystretch}{1.2}
\caption{Description of the System Parameters}
\label{table_1}
\centering
\begin{tabular}{c|l}\hline \hline
 \bfseries Parameter & \bfseries  Description\\ \hline 
$x_i \: (rad) $ & Angular position of $i^{\text{th}}$ joint\\ 
$\dot{x}_i \: (rad/s)$ & Angular velocity of $i^{\text{th}}$ joint \\ 
$u_i \: (N.m)$ & Applied torque at $i^{\text{th}}$ joint\\ 
$m_i \: (kg)$ & Mass of $i^{\text{th}}$ link\\ 
$l_i \: (m)$ & Length of $i^{\text{th}}$ link\\ 
$F_{s_i}\:(kg.m^2/s^2)$ & Static friction coefficient of $i^{\text{th}}$ joint\\ 
$F_{v_i}\:(kg.m^2/s)$ & Viscous friction coefficient of $i^{\text{th}}$ joint \\ 
$g \:(m/s^2)$ & Gravity acceleration  \\ \hline \hline
\end{tabular}
\end{table}
\begin{table}[!t]
\renewcommand{\arraystretch}{1.2}
\caption{System Parameters with Specified Values}
\label{table_2}
\centering
\begin{tabular}{c|c|c|c|c|c|c|c|c} \hline \hline 
 $m_1$ & $m_2$ & $l_1$ & $l_2$ & $F_{s_1}$ & $F_{s_2}$ & $F_{v_1}$ & $F_{v_2}$ & $g$\\ \hline
 7.0 & 4.0 & 0.5 & 0.5 & 0.8 & 0.8 & 4.0 & 4.0 & 9.8\\ \hline \hline
\end{tabular}
\end{table}
We considered 50 percent parametric uncertainty for the nominal values given in Table \ref{table_2} to generate the data and train the CGAN and showed the results of training the CGAN in Fig. \ref{fig:F_CGAN}.
According to Table \ref{table_3}, we chose the controller parameters for our proposed schemes and the baseline controller \cite{Yang2015}, where $N_p$ denotes the number of the policy's tunable parameters.
\begin{table}[!t]
\renewcommand{\arraystretch}{1.2}
\caption{Controllers Parameter Values}
\label{table_3}
\centering
\begin{tabular}{c|c|c|c|c|c} \hline \hline 
 Parameter & $\beta$ & $\alpha$ & $k$ & $\lambda_0$ & $N_p$  \\ \hline 
Proposed Controllers & 0 & 5 & 20 & diag([4,\,2])&2\\ \hline 
Baseline Controller & Adaptive & 5 & 20 & diag([4,\,2])&31\\ \hline  \hline 
\end{tabular}
\end{table}
According to \eqref{eq4.3} and \eqref{eq4.6}, the closed-loop extended tracking error dynamics depend on $\dot{u}$. Therefore, we can add a bias to the control input applied to the system. To bring the initial control input to zero, we consider a new control input as $u_{new}\triangleq u(t)-u(0)$, such that $u(0)=ke(0)-\hat{C}_{u_{ss}}(0)\Phi_\omega^{n_u}(0)$.
The initial states of the system are selected as $x(0)=[0,0]^T$. 
Our objective is to control the robot arm such that it tracks optimally the desired trajectory $x_d(t)=[\frac{\pi}{4}+0.2\cos(2\pi t),\sin(2\pi t)]^T$ in presence of the uncertainty and disturbance.

We have run two simulations to evaluate and compare the controllers' performances. The former has been conducted without the disturbance $d$. Figs. \ref{fig:F_Result_x} and \ref{fig:F_Result_u} show the results.
While after the 22nd second of this simulation, we suddenly load the robot with 25 percent of the second link's mass.
Therefore, a rapid uncertainty change occurs in the system, and our controllers can recover the system appropriately. 
In the second simulation, we subject the system to a bounded disturbance whose amplitude is about 50 percent of the required input torque as
\begin{align}\label{eq5.2}
d=[-80\sin(10t),\: 30\cos(10t)]^T.
\end{align}

\begin{table*}[!t]
\begin{threeparttable}[b]
\renewcommand{\arraystretch}{1.2}
\caption{Performance Measurement}
\label{table_4}
\centering
\setcellgapes{1.3pt}\makegapedcells
\begin{tabular}{|c||ccc|ccc|}\cline{2-7}\hline \hline 
\multicolumn{1}{|c||}{\multirow{2}{*}{Performance Criteria}}&\multicolumn{3}{c|}{First Simulation (With Varying Uncertainty)}&\multicolumn{3}{c|}{Second Simulation (With Disturbance)}\\ \cline{2-7}
\multicolumn{1}{|c||}{}&\makecell{Baseline Controller\\(Adaptive $\beta$)}&\makecell{Indirect Approach\\($\beta=0$)}&\makecell{Direct Approach\\($\beta=0$)} &\makecell{Baseline Controller\\(Adaptive $\beta$)}&\makecell{Indirect Approach\tnote{a}\\($\beta=0$)}&\makecell{Direct Approach\\($\beta=0$)} \\ \hline \hline
IAE &             16.43  & 11.50  &  \textbf{11.25}   &  22.77    &  19.46  &  \textbf{15.33}\\ \hline 
CE &               30.55 & 29.20  &  \textbf{28.55}   &  36.15    &  32.70  &  \textbf{31.08}\\ \hline 
IAR &             49.25  & 30.5   &   \textbf{29.64}  &  63.74   &  44.06  &  \textbf{40.79}\\ \hline 
Cost$(J_p)$ & 96.23  &  71.2  &   \textbf{69.44}     & 122.66  &  96.22  &  \textbf{87.20}\\ \hline \hline
\end{tabular}

 \begin{tablenotes}
    \item [a] In the presence of the disturbance, EKF can't converge. 
   \end{tablenotes}
  \end{threeparttable}
\end{table*}

The performance of the closed-loop system for different controllers and situations is measured in terms of tracking error, control effort, and actuation rate by the following cost functions and reported in Table \ref{table_4}. 
\begin{align}\label{eq5.1}
\text{Integral of Absolute Error (IAE)}=& \int_0^{t_f}  Q\|\tilde{x}(\tau)\|_1 d\tau, \nonumber\\
\text{Control Effort (CE)}=& \int_0^{t_f} R_1\|u(\tau)\|_1 d\tau, \nonumber\\
\text{Integral of Actuation Rate (IAR)}=& \int_0^{t_f} R_2\|\dot{u}(\tau)\|_1 d\tau, \nonumber\\ 
\text{Cost:}\; J_p \triangleq \text{IAE} + \text{CE} &+\text{IAR}. 
\end{align}
where $t_f$ is the final time of the simulation; $Q$, $R_1$, and $R_2$ are constant appropriate positive diagonal matrices. 

Because we selected control gain $\beta$ equals zero for our proposed methods, we do not have any high-frequency oscillation in the control input (Fig. \ref{fig:F_Result_u}). 
Therefore, according to Table \ref{table_4}, our proposed methods yield a minimum actuation rate. 
Based on \eqref{eq4.32}, there will be an ultimate bound for tracking error that affects performance. Still, the adaptation mechanisms can tune the generator network's input noise to reduce this ultimate bound by reducing $\delta_3$ in \eqref{eq4.3333} and \eqref{eq4.35}. 
Also, the RL framework can optimize performance in the direct approach by balancing the control effort and the tracking error. 
While in the baseline controller \cite{Yang2015} and \cite{Fan2018b}, the control gain $\beta$ must be adjusted by the tracking error adaptively. Therefore, if a significant tracking error occurs for the baseline controller, the value of $\beta$ increases and chattering in the control input worsens (Fig. \ref{fig:F_Result_u}).

\begin{figure}[!h]
  \centering
  \includegraphics[width=1\columnwidth]{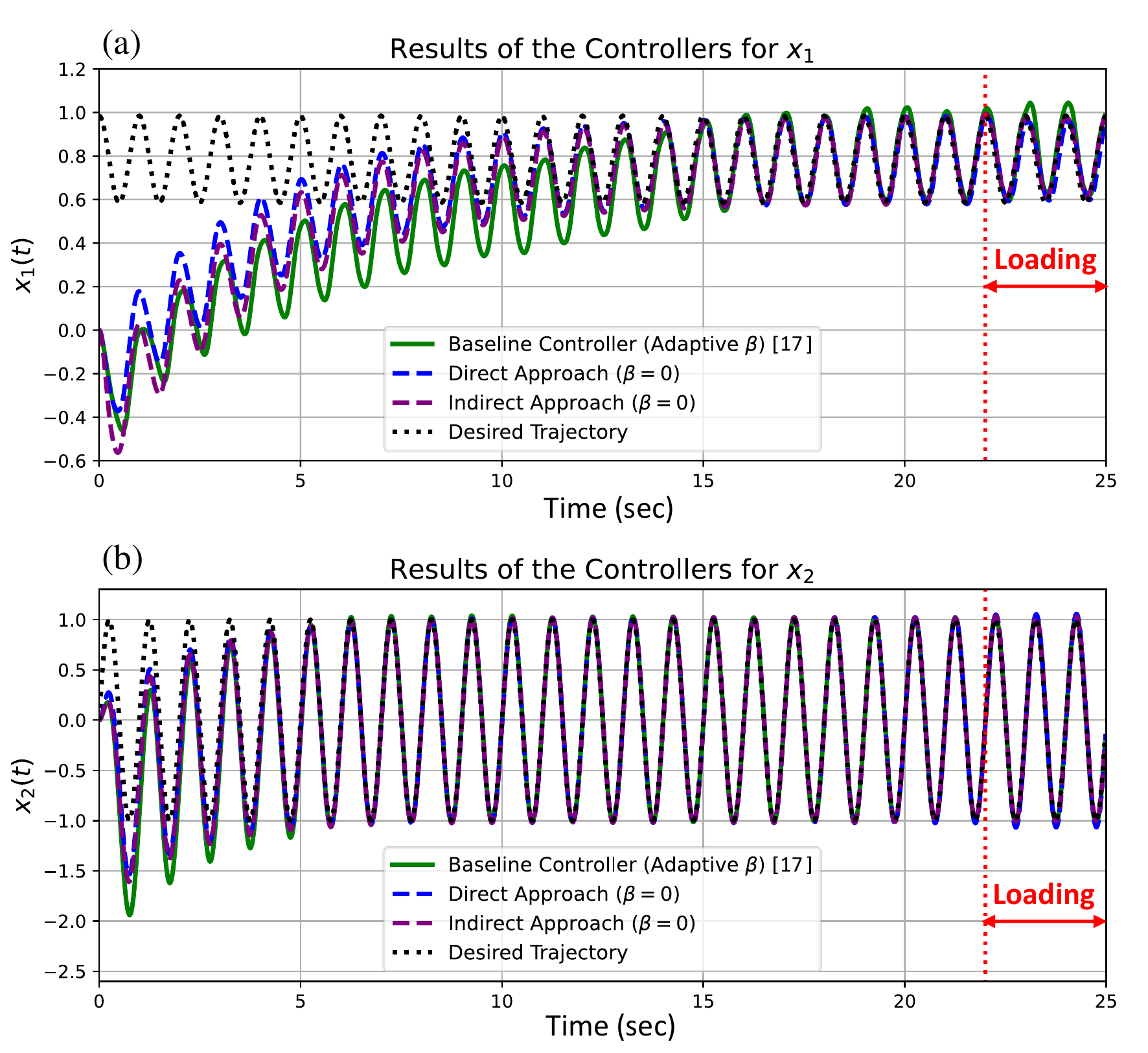}
  \caption{Position tracking of the first simulation for baseline controller \cite{Yang2015} with adaptive $\beta$ (solid green line), direct (dashed blue line) and indirect (dashed purple line) adaptive approach with $\beta=0$: (a) angular position of the first joint and (b) the second joint. The robot arm is loaded after the 22nd second.}\label{fig:F_Result_x}
\end{figure}
\begin{figure}[!h]
  \centering
  \includegraphics[width=1\columnwidth]{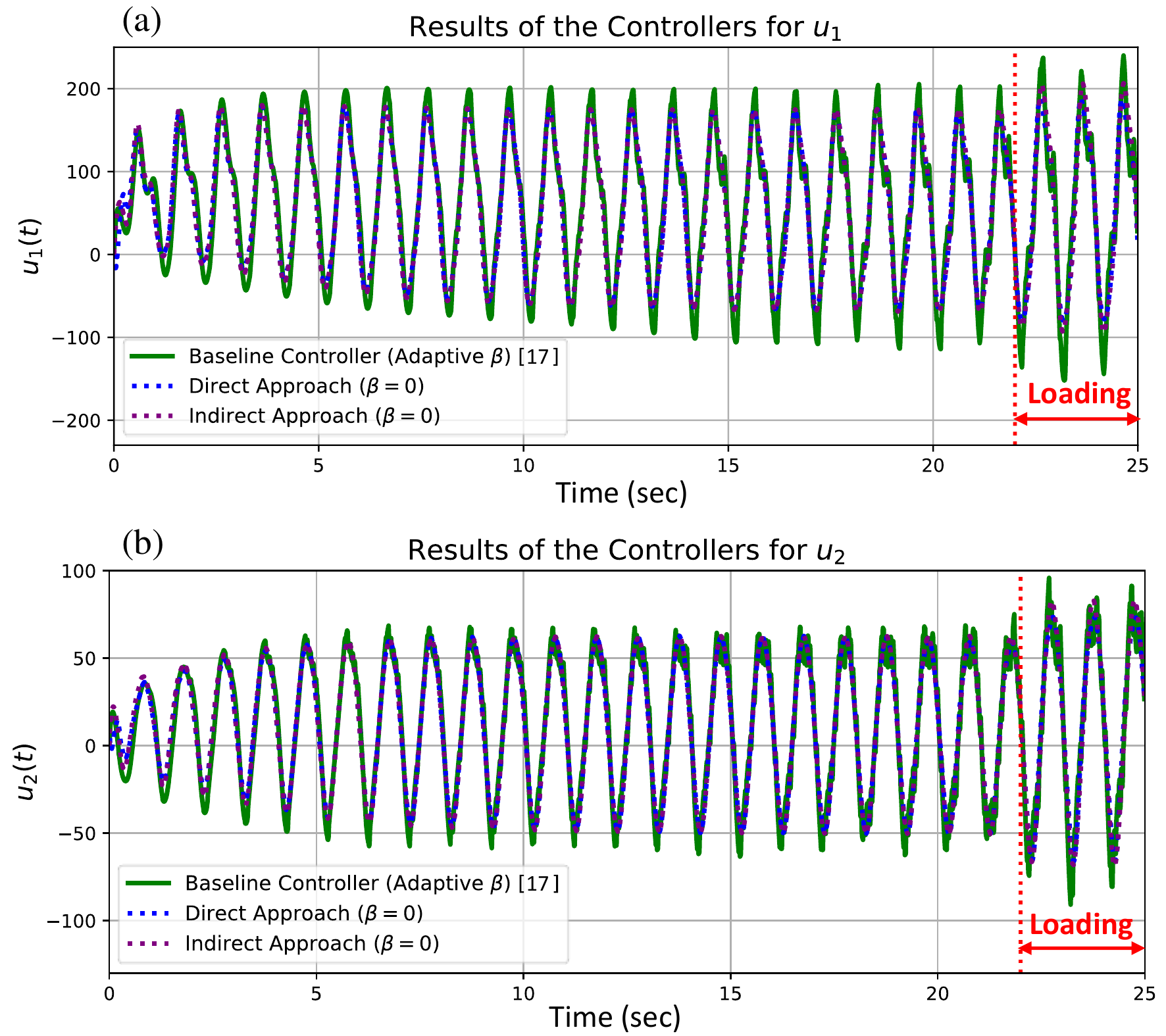}
  \caption{Control input signal of the first simulation for baseline controller \cite{Yang2015} with adaptive $\beta$ (solid green line), direct (dotted blue line) and indirect (dotted purple line) adaptive approach with $\beta=0$: (a) applied torque at the first joint and (b) the second joint. The robot arm is loaded after the 22nd second.}\label{fig:F_Result_u}
\end{figure}

The results in Table \ref{table_4} demonstrate that all controllers can tackle the bounded disturbance and the varying parametric uncertainty and preserve the stability of the closed-loop system. However, our direct adaptive approach has more optimal results than the others since it leverages the RL framework to optimize the value function. Furthermore, because EKF needs to know the structure of the dynamic model to converge, it cannot converge in the presence of the disturbance. Therefore, the indirect adaptive approach does not yield optimal results when faced with the disturbance. However, it keeps the system stable and has a more optimal performance than the baseline controller because of the capability of the CGAN that consistently generates the patterns close to the system's behavior. 
On the other hand, the direct adaptive approach is highly efficient in such a significant disturbance \eqref{eq5.2}; however, we did not consider any disturbance similar to \eqref{eq5.2} in the training data set of CGAN. In other words, the conditional generator in the Meta-RL framework could produce an adequate pattern according to this disturbance without the need for the integral of the sign of error term. This is the consequence of an effective combination between the generator network's generalization and the RL framework's optimality.

Another advantage of the proposed control schemes compared to previous methods is the small number of policy's tunable parameters, as stated in Table \ref{table_3}. We have only two parameters ($l_z=2$) to adjust policy in these simulations, whereas the dynamic model \eqref{eq5.1111} has eight parametric uncertainties.
This dimension reduction of the parametric uncertainties makes the adaptation mechanisms convenient and fast.
Thus, our methods are suitable for fast and optimal adaptive control of dynamic systems with many unknown parameters, particularly large-scale systems.

%%%%%%%%%%%%%%%%%%%%%%%%%%%%%%%%%%%%%%%%%%%%%%%%%%%%%%%%%%%%%%%%%%%%%%%%%%%%%%%%%%%%%%%%%%%%%%%%%%%%%%

\subsection{Large-Scale Musculoskeletal System}
This subsection presents the motion control of the human lumbar spine (Fig. \ref{fig:F_Lumbar}) as a large-scale musculoskeletal system because this dynamic system has a high degree of freedom and many redundancies \cite{Christophy2012, Rupp2015} and can show the power of our proposed control scheme to face larger systems. The human vertebral column or spine consists of different parts and tissues, including intervertebral joints, discs, ligaments, tendons, and muscles. The muscles are active tissues and play the role of the actuators for the motion control system. The other tissues are passive. The spine's vertebrae together form a passive multi-rigid body \cite{Christophy2012, Rupp2015, Raabe2016}. Therefore, we can model intervertebral joints using the Euler-Lagrange equation. Also, we can model intervertebral discs and ligaments based on their viscoelasticity property \cite{Christophy2012, Rupp2015, Raabe2016}. We consider a rigid-tendon model \cite{Millard2013} for the motion control of the human spine as follows: 
\begin{align} 
M(q) \ddot{q} + C(q,\dot{q}) 	&= \tau_p(q,\dot{q}) + \tau_{load} + \tau_M (q,\dot{q},u), \label{eq:MRB}\\ 
\tau_M (q,\dot{q},u)  &= R(q) F_M(q,\dot{q},u), \label{eq:ARM}\\
 F_M(q,\dot{q},u) &= F_A(q,\dot{q}) u + F_P(q,\dot{q}), \label{eq:Muscle}
\end{align}
where $q\in\mathbb{R}^{n\times 1}$ is the generalized coordinates for description of the system configuration; $M(.)\in\mathbb{R}^{n\times n}$ is the inertia matrix and is positive-definite; $C(.,.)\in\mathbb{R}^{n\times 1}$ is the torque related to the centrifugal, Coriolis, and gravity forces; $\tau_p\in\mathbb{R}^{n\times 1}$ is the passive torque originated from the passive tissues; $\tau_{load}\in\mathbb{R}^{n\times 1}$ is the external load torque generated on the spine; $\tau_M\in\mathbb{R}^{n\times 1}$, $F_M\in\mathbb{R}^{l_m\times 1}$, $F_A\in\mathbb{R}^{l_m\times 1}$, and $F_P\in\mathbb{R}^{l_m\times 1}$  are the total torque, total force, active force, and passive force caused by muscle fibers, respectively, and computed by models suggested by \cite{Millard2013}; $u\in\mathbb{R}^{l_m\times 1}$ is the activation level of muscle fibers consisting of real numbers between zero and one; $l_m$ is the number of the muscle fibers; $R(.)\in\mathbb{R}^{n \times l_m}$ is the moment arm of muscles related to the joint and muscle geometry. 
We exploit the data prepared by \cite{Raabe2016} in the OpenSim software \cite{Delp2007} for the geometry of the musculoskeletal system.
Also, we have chosen the required parameters and constants to simulate the musculoskeletal system from \cite{Rupp2015, Christophy2012}. 

Similar to \cite{Nasseroleslami2014}, we consider a 3D inverted pendulum with three rotational DoF to model the trunk of the vertebral column in (\ref{eq:MRB}). The orientation of this rigid body is obtained by three Euler angles for rotation around x, y, and z axes, and displayed by $q_1$, $q_2$, and $q_3$, respectively. 
Also, we consider 162 muscle fibers in seven groups shown in Fig. \ref{fig:F_Lumbar} (b), including 2 fibers for Rectus Abdominis (RA) muscle, 12 fibers for External Oblique (EO) muscle, 12 fibers for Internal Oblique (IO) muscle, 24 fibers for Iliocostalis (IL) muscle, 52 fibers for Longissimus (LT) muscle, 36 fibers for Quadratus (QL) muscle, and 24 fibers for Multifidus (MF) muscle. Two markers, shown in Fig. \ref{fig:F_Lumbar} (a), were considered to capture the system's movement in the workspace. 
\begin{figure}[!h]
  \centering
  \includegraphics[width=1\columnwidth]{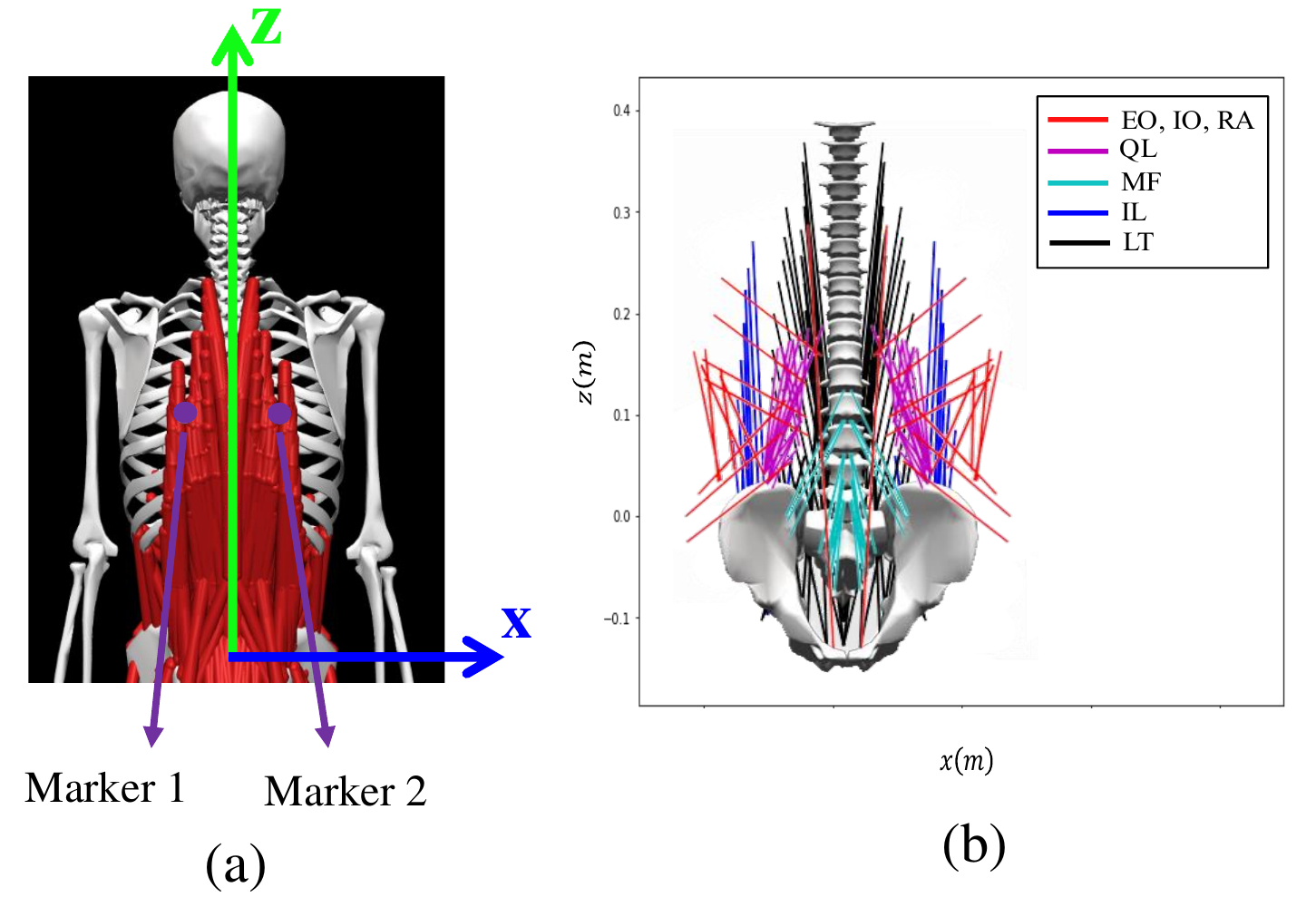}
  \caption{Musculoskeletal system of the spine in the x-z plane: (a) shows the human spine and location of the considered markers for motion capture, (b) position of modeled muscle fibers around the spine.}\label{fig:F_Lumbar}
\end{figure}

The goal is to determine the muscle activation vector ($u$) as system input at any time so that the markers on the human spine track the desired movement. Furthermore, the controller must be able to handle the redundancies in the system according to some criteria, tackle the unknown dynamics, disturbances, and perturbations optimally, and preserve the stability of the closed-loop system. Hence, we apply the proposed Mete-RL framework with a trajectory planning and the computed muscle control (CMC) method introduced by \cite{Thelen2003} to control the system (see Supplementary Material Section III for details). 

To evaluate the proposed Meta-RL framework, we investigate the repetitive rotation of the spine around the x-axis, called flexion/extension movement, while the external load on the spine varies. The spine bends forward up to 50 degrees in the y-z plane and returns to its start point. This motion is repeated with a period of 1.5 seconds. After the 8th second, an unknown load is applied to the spine. 
Figs. \ref{fig:F_FE_q}, \ref{fig:F_FE_Flex}, and \ref{fig:F_FE_Ext} show the orientation of the spine, the flexor muscles' activation level, and the extensor muscles' activation level, respectively.
As can be seen from the results, after applying the unknown load on the spine, the control structure can recover the system well and adapt to the change in the system, and after a short time, the system tracking error is zero. In addition, by loading the spine, the torque required by the system increases; consequently, the activation of the muscles also increases. 
\begin{figure}[!h]
  \centering
  \includegraphics[width=0.9\columnwidth]{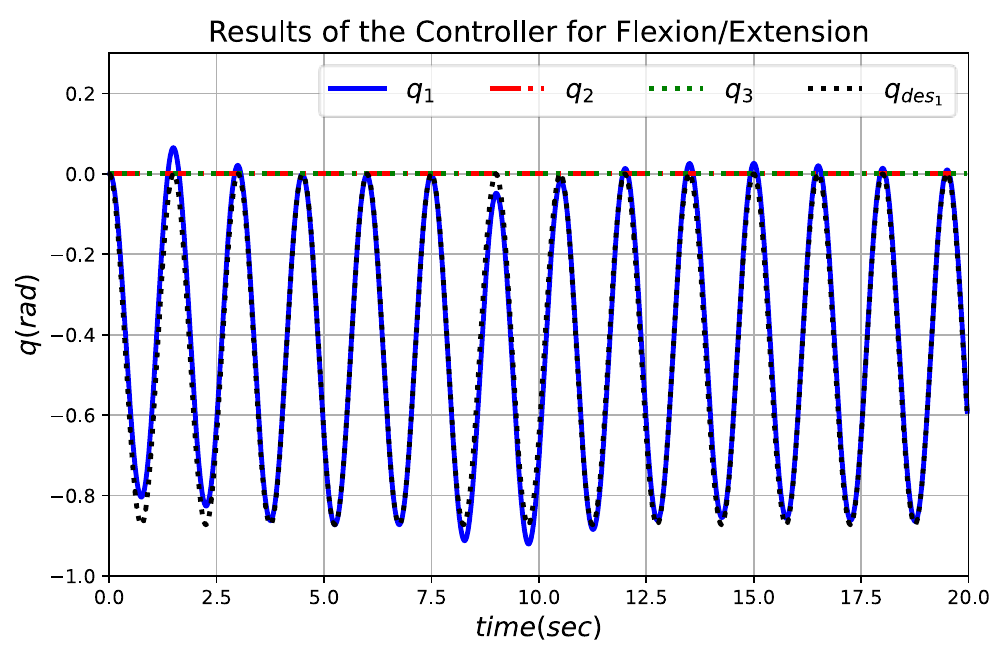}
  \caption{Orientation angles of the spine for flexion/extension movement in the y-z plane. The spine is loaded after the 8th second.}\label{fig:F_FE_q}
\end{figure}
\begin{figure}[!h]
  \centering
  \includegraphics[width=1\columnwidth]{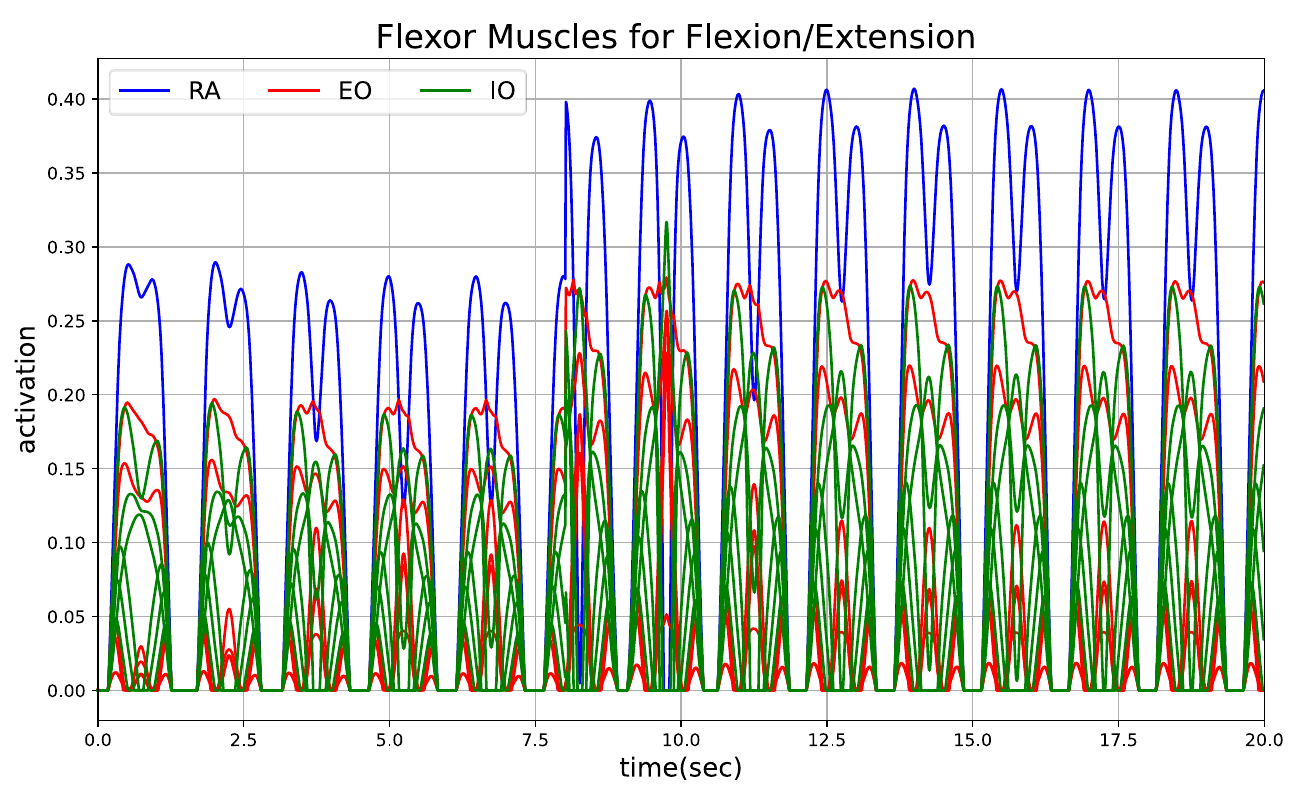}
  \caption{The activation level of the flexor muscles on the right-hand side of the spine for flexion/extension movement in the y-z plane. The spine is loaded after the 8th second.}\label{fig:F_FE_Flex}
\end{figure}
\begin{figure}[!h]
  \centering
  \includegraphics[width=1\columnwidth]{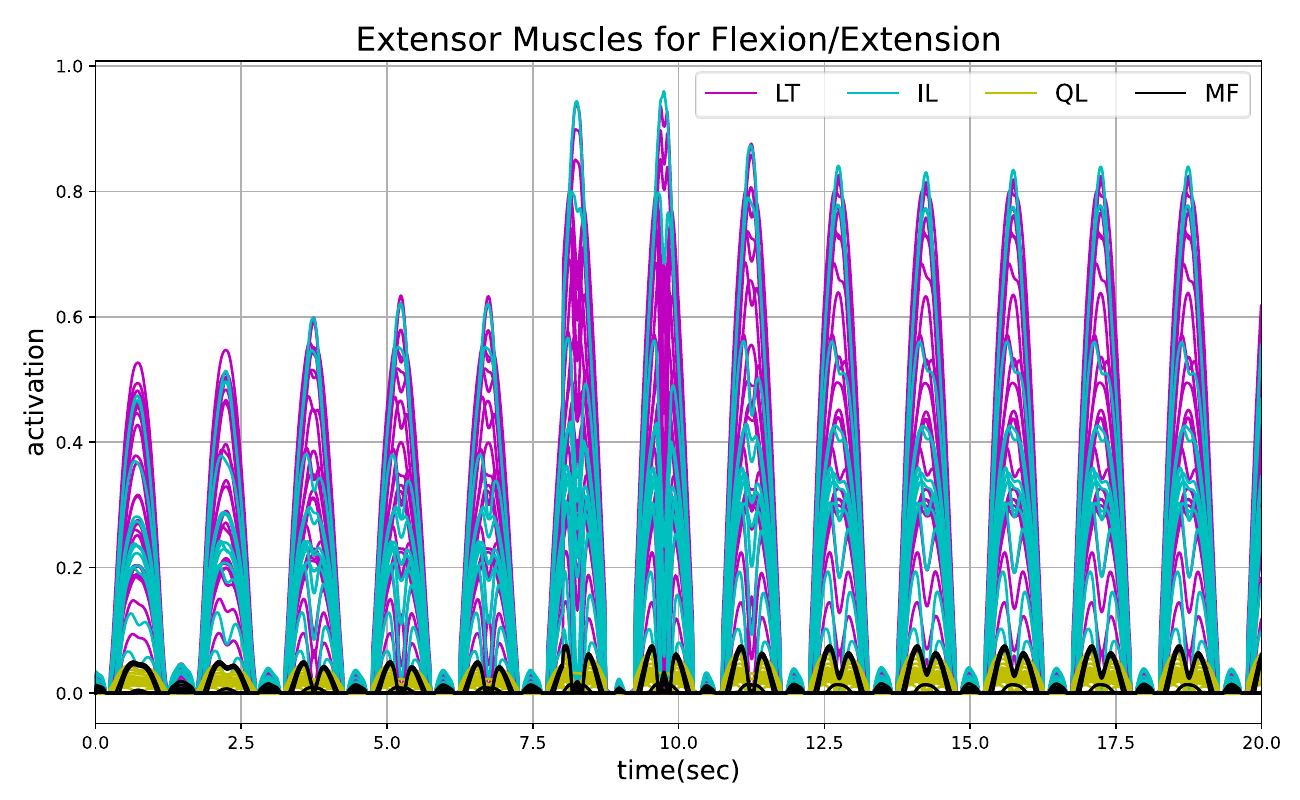}
  \caption{The activation level of the extensor muscles on the right-hand side of the spine for flexion/extension movement in the y-z plane. The spine is loaded after the 8th second.}\label{fig:F_FE_Ext}
\end{figure}

Furthermore, we have prepared two other simulation results for point-to-point oscillatory (P2POSC) and point-to-point regulatory (P2PREG) movements in Section IV of Supplementary Material.

%%%%%%%%%%%%%%%%%%%%%%%%%%%%%%%%%%%%%%%%%%%%%%%%%%%%%%%%%%%%%%%%%%%%%%%%%%%%%%%%%%%%%%%%%%%%%%%%%%%%%%

\section{Conclusion} \label{section: Conclusion and Future work}

The main goal of the current study was to design fast and optimal adaptive tracking control for a class of MIMO uncertain nonlinear systems using RL and CGAN that leads to a novel Meta-RL framework. 
In this data-driven control scheme, CGAN was used to model the uncertain dynamic system using previous data, RL was employed to adapt the CGAN to the system optimally, and the RISE controller ensures stability and robustness. 
Data-driven learning-based control makes control tasks more flexible and scalable. However, stability analysis is a significant challenge for such approaches. This paper solves this research gap by leveraging the RISE controller.
We used CGAN to reduce the dimension of the uncertainties and produce the adaptation mechanism more efficiently than common adaptive controller and policy gradient methods.
The CGAN, as mentioned above, was also introduced and exploited as an extended describing function, adaptable CPG, and fault detector. 
Finally, the simulation results showed the proposed schemes' effectiveness and scalability.
In comparison to RISE-NN structures, our control schemes can reveal more optimal performance without the need for the integral of the sign of error term.

Future efforts will focus on broadening the class of nonlinear systems to be controlled. Also, it would be interesting to handle some hard constraints on the input or state variables in the proposed control structures. Finally, we need to establish a control scheme that simultaneously learns CGAN and Meta-RL framework. 
Moreover, to increase the generalization of the CGAN, we need a large amount of data, which is another limitation for practical applications.

\ifCLASSOPTIONcaptionsoff
  \newpage
\fi

% trigger a \newpage just before the given reference
% number - used to balance the columns on the last page
% adjust value as needed - may need to be readjusted if
% the document is modified later
%\IEEEtriggeratref{8}
% The "triggered" command can be changed if desired:
%\IEEEtriggercmd{\enlargethispage{-5in}}

% references section

% can use a bibliography generated by BibTeX as a .bbl file
% BibTeX documentation can be easily obtained at:
% http://mirror.ctan.org/biblio/bibtex/contrib/doc/
% The IEEEtran BibTeX style support page is at:
% http://www.michaelshell.org/tex/ieeetran/bibtex/
%\bibliographystyle{IEEEtran}
% argument is your BibTeX string definitions and bibliography database(s)
%\bibliography{IEEEabrv,../bib/paper}
%
% <OR> manually copy in the resultant .bbl file
% set second argument of \begin to the number of references
% (used to reserve space for the reference number labels box)

\bibliographystyle{IEEEtran}
%\bibliography{IEEEabrv,My_Collection}
\bibliography{IEEEabrv,My_Collection_New}
%\bibliography{IEEEfull,My_Collection_New}

% biography section
% 
% If you have an EPS/PDF photo (graphicx package needed) extra braces are
% needed around the contents of the optional argument to biography to prevent
% the LaTeX parser from getting confused when it sees the complicated
% \includegraphics command within an optional argument. (You could create
% your own custom macro containing the \includegraphics command to make things
% simpler here.)
%\begin{IEEEbiography}[{\includegraphics[width=1in,height=1.25in,clip,keepaspectratio]{mshell}}]{Michael Shell}
% or if you just want to reserve a space for a photo:

% You can push biographies down or up by placing
% a \vfill before or after them. The appropriate
% use of \vfill depends on what kind of text is
% on the last page and whether or not the columns
% are being equalized.

%\vfill

% Can be used to pull up biographies so that the bottom of the last one
% is flush with the other column.
%\enlargethispage{-5in}

% that's all folks
\end{document}

\def\theequation{S\arabic{equation}}
\def\thefigure{S\arabic{figure}}
%
% paper title
% Titles are generally capitalized except for words such as a, an, and, as,
% at, but, by, for, in, nor, of, on, or, the, to and up, which are usually
% not capitalized unless they are the first or last word of the title.
% Linebreaks \\ can be used within to get better formatting as desired.
% Do not put math or special symbols in the title.

\title{Supplementary Material for Fast and Optimal Adaptive Tracking Control: A Novel Meta-Reinforcement Learning via \\ Conditional Generative Adversarial Net}

\author{Mohammad~Mahmoudi,
            Nasser~Sadati, \IEEEmembership{Senior Member,~IEEE}% <-this % stops a space           
\thanks{M. Mahmoudi and N. Sadati are with the Department of Electrical Engineering, Sharif University of Technology, Tehran, Iran (e-mails: m.mahmoudi.f75@gmail.com, sadati@sharif.edu). }}%
% The paper headers
\markboth{Submitted to IEEE Transactions on}%
{Mahmoudi \MakeLowercase{\textit{et al.}}:Fast and Optimal Adaptive Tracking Control: A Novel Meta-Reinforcement Learning via Conditional Generative Adversarial Net}%
% make the title area
\maketitle

\IEEEpeerreviewmaketitle

\section{Proof of Lemma 1}

\begin{proof}
Similar to \cite{Yang2015,Xian2004}, by integrating and expanding (52) and using the integration by parts, we have 
\begin{align}\label{eq4.21}
\Omega(t)&=\Omega(0) - \int_0^t \alpha e^T(\tau) \left[\dot{\delta}_3(\tau)+\dot{D}(\tau)-\beta_d sgn(e(\tau)) \right. \nonumber \\
&\left. +\frac{1}{\alpha}\left(\ddot{\delta}_3(\tau)+\ddot{D}(\tau)+M_2(\tau)-\dot{M}_1(\tau)\right)\right] d\tau -\left[ e^T(\tau) \right.\nonumber \\
&\left. \left(\dot{\delta}_3(\tau)+\dot{D}(\tau)+M_1(\tau)\right)\right]^t_0 +[\beta_d \|e(\tau)\|_1]^t_0,
\end{align}
After applying (53) and using the Cauchy-Schwarz inequality,
\begin{align}\label{eq4.21}
&\Omega(t)\ge \int_0^t \alpha e^T(\tau) \left[\beta_d-\|\dot{\delta}_3(\tau)\|_1-\|\dot{D}(\tau)\|_1 \right. \nonumber \\ 
&\left. -\frac{1}{\alpha}\left(\|\ddot{\delta}_3(\tau)\|_1  +\|\ddot{D}(\tau)\|_1+\|M_2(\tau)\|_1
+\|\dot{M}_1(\tau)\|_1\right)\right] d\tau \nonumber \\ 
&+\left[ \beta_d 
 -\|\dot{\delta}_3(t)\|_1-\|\dot{D}(t)\|_1-\|M_1(t)\|_1\right]\|e(t)\|_1.
\end{align}
Now, if $\beta_d$ is chosen based on (54), then $\Omega(t) \ge 0$ holds. 
\end{proof}

\section{Proof of Theorem 1}
\begin{proof}
To show the stability of the closed-loop system, we consider the mentioned positive-definite Lyapunov function (55).
By differentiating (55) and using (33), (39), (42), (43), (45), and (52), after simplification, we obtain
\begin{align}\label{eq4.23}
\dot{L} &= -\alpha \|e\|^2 - \|s\|^2 - k \|s\|^2+s^T\tilde{N}-2\gamma^2\tilde{v}^T \Pi \tilde{v} \nonumber\\
&+2\gamma \tilde{v}^T \left[\sum_{i=0}^p a(t_i)\left(r(t_i)+ \dot{a}^T(t_i)v(t)- \gamma a^T(t_i)v^*(t) \right)\right]\nonumber \\
&+2\gamma \tilde{v}^T a(t) \delta_{TD}(t)+s^T(\beta-\beta_d)sgn(e).
\end{align}
By using Remark 3, (40) and (51) and
employing the Cauchy-Schwarz inequality,
we can obtain an upper bound on \eqref{eq4.23} as follows. Then, by adding and subtracting $\delta^2_{TD}(t)$ and $\varepsilon\|\xi\|^2$ into the right-hand side of the result, we have
\begin{align}\label{eq4.24}
\dot{L} & \le -\bar{\alpha}\|\xi\|^2 +\left[\|s\|_2 \rho (\|\xi\|) \|\xi\|- k \|s\|^2\right]+|\beta-\beta_d|\|s\|\nonumber\\
&+ 2 \gamma \|\tilde{v}\| \sum_{i=0}^p \|a(t_i)\| \left( \|\dot{a}(t_i)\| \|\tilde{v}\|+|\gamma \delta_4(t_i) - \dot{\delta}_4(t_i)| \right) \nonumber\\
&\pm \varepsilon\|\xi\|^2- \delta^2_{TD} -2\gamma^2\tilde{v}^T \Pi \tilde{v} + (2\gamma \tilde{v}^T a(t)+\delta_{TD}) \delta_{TD},
\end{align}
where $\varepsilon>0$, $\bar{\alpha}\triangleq\min\{\alpha,1\}$.
After completing the squares on the bracketed term and applying the difference of squares identity on the last term in \eqref{eq4.24}, by using (49), we obtain
\begin{align}\label{eq4.25}
\dot{L} & \le -\left[\sqrt{k} \|s\|- \frac{\rho (\|\xi\|) \|\xi\|}{2\sqrt{k}} \right]^2 +\frac{\rho^2 (\|\xi\|) \|\xi\|^2}{4 k}\nonumber\\
& -(\bar{\alpha}-\varepsilon)\|\xi\|^2+\left(|\beta-\beta_d|-\varepsilon\|\xi\|\right)\|\xi\|- \delta^2_{TD} \nonumber\\
&-2\gamma^2 \lambda_{\min}(\Pi) \|\tilde{v}\|^2 + 2 \gamma p \dot{\bar{a}} \|\tilde{v}\|^2 +2 \gamma p(\gamma \bar{\delta}_4 + \dot{\bar{\delta}}_4) \|\tilde{v}\| \nonumber\\
&-\gamma^2 \|\tilde{v}^Ta\|^2+\left(r(t)+ \dot{a}^T(t)v(t)- \gamma a^T(t_i)v^*(t) \right)^2 \nonumber\\
&\le -\left[\sqrt{k} \|s\|- \frac{\rho (\|\xi\|) \|\xi\|}{2\sqrt{k}} \right]^2 -\left(\bar{\alpha}-\varepsilon-\frac{\rho^2 (\|\xi\|)}{4 k}\right)\|\xi\|_2^2 \nonumber\\
&-\left(\gamma^2 \|\tilde{v}^Ta\|^2-(\gamma \bar{\delta}_4 + \dot{\bar{\delta}}_4)^2\right) \nonumber\\
&-\left(\left(2\gamma^2 \lambda_{\min}(\Pi) -2 \gamma p \dot{\bar{a}}-\dot{\bar{a}}\right)\|\tilde{v}\|-2 \gamma p(\gamma \bar{\delta}_4 + \dot{\bar{\delta}}_4)\right)\|\tilde{v}\| \nonumber\\
&- \delta^2_{TD}-\left(\varepsilon\|\xi\|-|\beta-\beta_d|\right)\|\xi\|.
\end{align}
If the following inequalities are satisfied, then the right-hand side of the inequality \eqref{eq4.25} is negative, and the upper bound for the derivative of Lyapunov function (55) is given by $\dot{L}(\mu,t) \le -W(\xi,\tilde{v})$, such that $W(.)$ is a continuous positive function. Therefore, all system signals are bounded based on Barbalat's lemma and Theorem 8.4 of \cite{Khalil2002}.
\begin{gather}
\frac{|\beta-\beta_d|}{\varepsilon} \le \|\xi\| \le \rho^{-1}\left(2\sqrt{k(\alpha_3-\varepsilon)}\right)\label{eq4.26},\\
0 < \varepsilon < \bar{\alpha} \le 1, \label{eq4.27}\\
\|\tilde{v}^Ta\| \ge \frac{\gamma \bar{\delta}_4 + \dot{\bar{\delta}}_4}{\gamma},\label{eq4.28}\\
\|\tilde{v}\| \ge \frac{2 \gamma p(\gamma \bar{\delta}_4 + \dot{\bar{\delta}}_4)}{2\gamma^2 \lambda_{\min}(\Pi) -2 \gamma p \dot{\bar{a}}-\dot{\bar{a}}},\label{eq4.29}\\
2\gamma^2 \lambda_{\min}(\Pi) -2 \gamma p \dot{\bar{a}}-\dot{\bar{a}} \ge 0. \label{eq4.30}
\end{gather}
Also, if $\beta=\beta_d$, then $\xi$ converges to zero asymptotically or $e(t)\rightarrow 0$ as $t\rightarrow \infty$. By treating $\tilde{x}(t)$ and $e(t)$ respectively as the output and the input of the stable linear system (18), subsequently, the tracking error $\tilde{x}(t)$ approaches zero asymptotically.
\end{proof}

\section{Proposed Meta-RL Framework for Musculoskeletal Systems}
We apply the proposed Mete-RL framework in the three next steps to determine the muscle activation vector ($u$) as system input at any time so that the markers on the musculoskeletal structure track the desired movement. 

1) we transform the desired movement for the markers in the workspace into the joint space using inverse kinematics \cite{Delp2007} and the following trajectory planning to obtain the TFS coefficient of the desired trajectory. We solve the subsequent trajectory planning by convex-concave programming introduced by \cite{Shen2016}.
\begin{align} 
C_{q_d}= \; & \underset{C_q}{\text{argmin}} {\quad trace\left(C_{q} \Delta C_{q}^T\right) }  \label{eq:cost_TP} \\
 \text{s. t.}  \quad &  \| C_{q} \Phi_\omega^{n_x} (\tau) - q_{obs} (\tau)\|_2^2 \geq d_{obs} ; \; \forall \tau =0,...,\tau_f   \label{eq:cons_obs_TP} \\
& q_{min} \leq C_{q} \Phi_\omega^{n_x} (\tau)  \leq q_{max} ;\; \forall \tau =0,...,\tau_f  \label{eq:cons_RoM_TP} \\
& C_{q} \Phi_\omega^{n_x} \left( j T+ \frac{T}{2} \right) = q_f ;\; \forall j =0,1,... \label{eq:cons_final_TP} \\
& C_{q} \Phi_\omega^{n_x} \left( j T \right) = q_0 ;\; \forall j =0,1,...  \label{eq:cons_start_TP} \\
& C_{q} I_s = 0_ {n \times n_x}, \label{eq:cons_sym_TP} 
\end{align}
where $\Phi_\omega^{n_x}(.)$ and $T$ are the same as (13) and (14); $q^{obs}$ represents the position of obstacles in the joint space; $d_{obs}$ is the minimum acceptable distance between system and obstacle in the joint space; $q_{min}$ and $q_{max}$ demonstrate the admissible biological range of motion (RoM) for the musculoskeletal system; $q_0$ and $q_f$ are the inverse kinematics of starting and final points for the markers, respectively; $\tau$ denotes the discrete time index; $\tau_f$ is the final time for the trajectory planning; $\Delta$ and $I_s$ are matrices as follows:
\begin{align} 
\Delta & \triangleq diag \left( [0, 1, 2^2, ..., n_x^2, 1, 2^2, ..., n_x^2]  \right) \nonumber \\
I_s & \triangleq \left[ 0_{n_x \times (n_x+1)} , I_{n_x \times n_x} \right]^T. \nonumber
\end{align}

The cost function (\ref{eq:cost_TP}) is equivalent to the integral of velocity w.r.t. time over a half-period. Also, we could use other matrices instead of $\Delta$ to minimize other criteria such as acceleration, rate of acceleration, or combination of them. 
We have defined constraints (\ref{eq:cons_obs_TP}) to prevent collisions with obstacles.
Constraints (\ref{eq:cons_RoM_TP}) are related to the admissible RoM. Constraints (\ref{eq:cons_final_TP}) and (\ref{eq:cons_start_TP}) force the musculoskeletal system to reach the final point in the middle of a period and return to the starting point at the end. Constraint (\ref{eq:cons_sym_TP}) has been defined to zero the TFS sine coefficients because of commutation on a similar path and symmetry. Trajectory planning described by (\ref{eq:cost_TP})-(\ref{eq:cons_sym_TP}) is for P2POSC motions. By choosing $j=0$ and ignoring constraint (\ref{eq:cons_sym_TP}), we will achieve P2PREG motions. 

Because the CGAN uses TFS coefficients as a label, we formulated the trajectory planning in the TFS space. Therefore, the trajectory planning output is well compatible with the CGAN. Additionally, the decision variables in the trajectory planning are reduced to the number of TFS coefficients of the desired trajectory; therefore, there is no need to find the desired trajectory at every time step. Hence, the computational cost is drastically reduced. 

2) We train the CGAN by recording data of the musculoskeletal system in the presence of different uncertainties. Therefore, this CGAN can be interpreted as a CPG located in the human nervous system.
Then, we give the TFS coefficient of the desired trajectory obtained by (\ref{eq:cost_TP})-(\ref{eq:cons_sym_TP}) as a label to the Meta-RL structure, demonstrated in Fig. 6, which produces the required torque for the motion control of the musculoskeletal system. 

3) We transform the required torque achieved by the Meta-RL structure into a desired muscle activation vector using the computed muscle control (CMC) method introduced by \cite{Thelen2003}. This method solves a static optimization problem to find an optimal activation level for muscles at any time so that the activation of muscles is minimized subject to muscles can produce the required torque. The CMC method needs to know the model of the muscles (70) and (71). If there is uncertainty in the muscle model, an error will occur in supplying the required torque by CMC. We can consider this error as $\delta_3$ mentioned in Assumption 5. Therefore, our control scheme is robust to the perturbations in the muscle model and completes the model-based functioning of the CMC method.

\section{More Simulation Results for the Musculoskeletal System of the Human Spine}

This section presents two other simulation results for the P2POSC and P2PREG movements of the human spine.

In the simulation of the P2POSC movement, the spine reaches the desired final point for the markers without any collision with an obstacle, then returns to the start point. This motion is repeated with a 5-second period. 
Figs. \ref{fig:F_OSC_Mark}, \ref{fig:F_OSC_Flex}, and \ref{fig:F_OSC_Ext} depict the position of the markers, the flexor muscles' activation level, and the extensor muscles' activation level, respectively. By leveraging the trajectory planning introduced by  (\ref{eq:cost_TP})-(\ref{eq:cons_sym_TP}), we generated the desired position for the markers to avoid obstacles.
As can be seen from the results of this simulation, the system encounters a slight tracking error when the spine reaches its final point. At this moment, the muscles cannot produce the torque required by the system, and the flexor muscles' activation level becomes saturated, as shown in Fig. \ref{fig:F_OSC_Flex}. 
Therefore, the system cannot reach this final point at this desired speed and acceleration. However, Fig. \ref{fig:F_OSC_Mark} indicates that the markers follow their desired value well, except at the endpoint. Therefore, the proposed controller could face the actuator saturation fault well. 
\begin{figure}[!h]
  \centering
  \includegraphics[width=1\columnwidth]{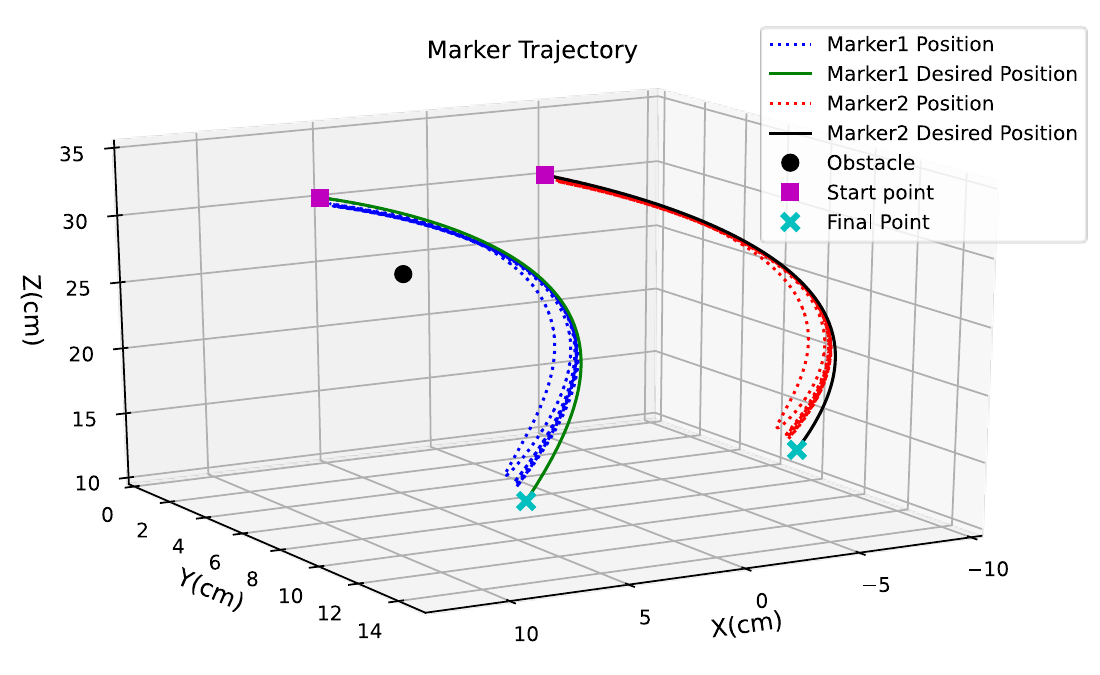}
  \caption{The position of the markers located on the spine in the 3D space.}\label{fig:F_OSC_Mark}
\end{figure}

\begin{figure}[!h]
  \centering
  \includegraphics[width=0.9\columnwidth]{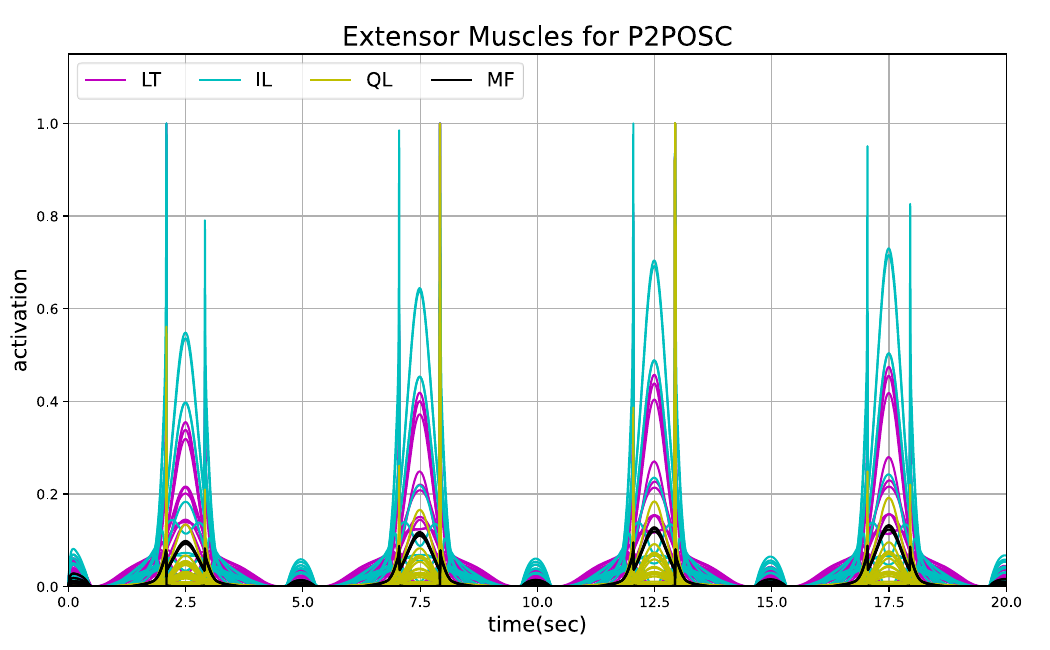}
  \caption{The activation level of the extensor muscles on the right-hand side of the spine for the P2POSC movement.}\label{fig:F_OSC_Ext}
\end{figure}

\begin{figure}[!h]
  \centering
  \includegraphics[width=0.9\columnwidth]{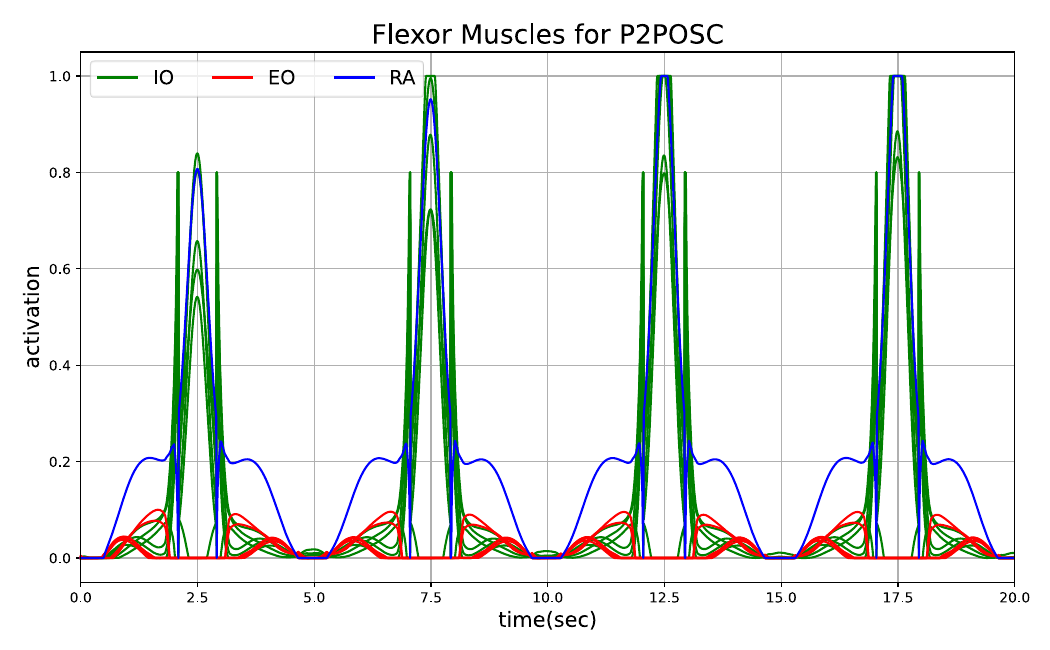}
  \caption{The activation level of the flexor muscles on the right-hand side of the spine for the P2POSC movement. Flexor muscles become saturated at the final point of the movement.}\label{fig:F_OSC_Flex}
\end{figure}

We have conducted the other simulation to show the ability of the proposed controller to regulate P2PREG movement. In this simulation, the spine reaches a desired final point for the markers without any collision with an obstacle and stays at that point ($q_f$) forever. Fig. \ref{fig:F_REG} shows the orientation angles of the spine during this simulation. By leveraging the trajectory planning introduced by  (\ref{eq:cost_TP})-(\ref{eq:cons_sym_TP}), $q_{des}$ is generated to avoid obstacles.
\begin{figure}[!h]
  \centering
  \includegraphics[width=0.9\columnwidth]{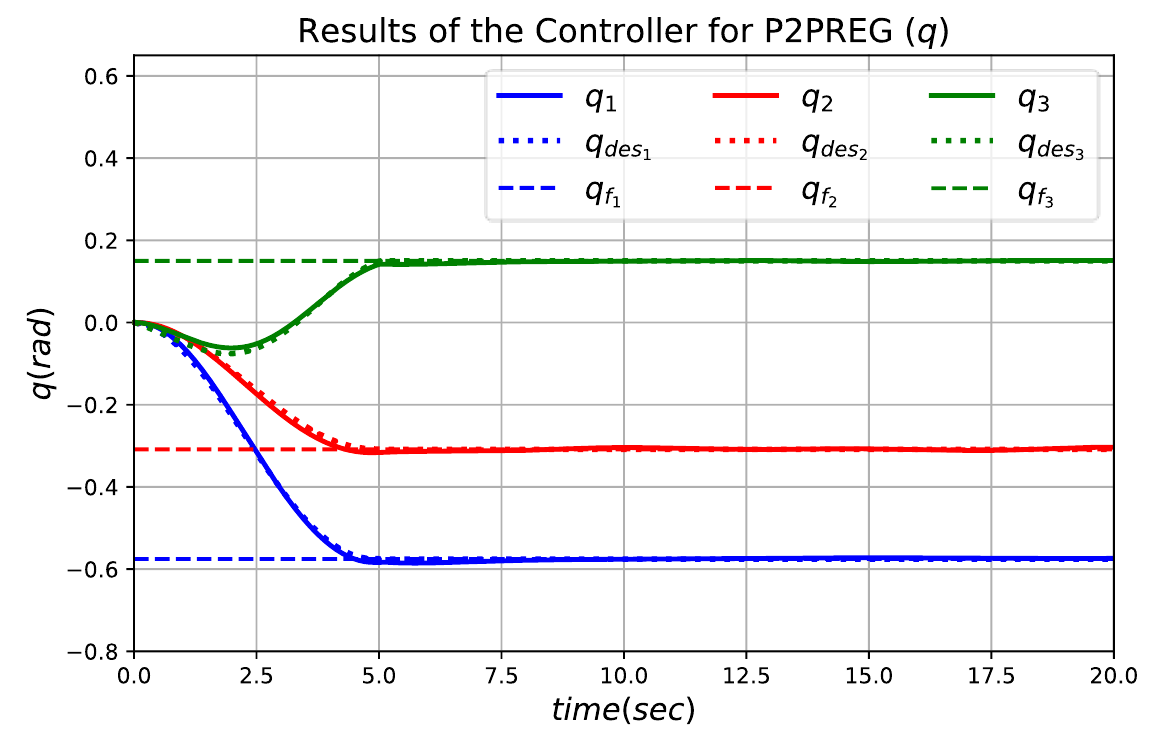}
  \caption{Orientation angles of the spine for P2PREG movement. The spine is fixed after the 5th second in the final point. Dashed lines and dotted curves show the final point and the desired trajectory, respectively.}\label{fig:F_REG}
\end{figure}

In all simulation results, we considered many parametric uncertainties and unmodeled dynamics in modeling biological tissues \cite{Rupp2015}. Therefore, we can use our Meta-RL framework for optimal adaptive control of large-scale systems.

\ifCLASSOPTIONcaptionsoff
  \newpage
\fi

% trigger a \newpage just before the given reference
% number - used to balance the columns on the last page
% adjust value as needed - may need to be readjusted if
% the document is modified later
%\IEEEtriggeratref{8}
% The "triggered" command can be changed if desired:
%\IEEEtriggercmd{\enlargethispage{-5in}}

% references section

% can use a bibliography generated by BibTeX as a .bbl file
% BibTeX documentation can be easily obtained at:
% http://mirror.ctan.org/biblio/bibtex/contrib/doc/
% The IEEEtran BibTeX style support page is at:
% http://www.michaelshell.org/tex/ieeetran/bibtex/
%\bibliographystyle{IEEEtran}
% argument is your BibTeX string definitions and bibliography database(s)
%\bibliography{IEEEabrv,../bib/paper}
%
% <OR> manually copy in the resultant .bbl file
% set second argument of \begin to the number of references
% (used to reserve space for the reference number labels box)
\bibliographystyle{IEEEtran}
%\bibliography{IEEEabrv,My_Collection}
\bibliography{IEEEabrv,My_Collection_New}
%\bibliography{IEEEfull,My_Collection_New}

% biography section
% 
% If you have an EPS/PDF photo (graphicx package needed) extra braces are
% needed around the contents of the optional argument to biography to prevent
% the LaTeX parser from getting confused when it sees the complicated
% \includegraphics command within an optional argument. (You could create
% your own custom macro containing the \includegraphics command to make things
% simpler here.)
%\begin{IEEEbiography}[{\includegraphics[width=1in,height=1.25in,clip,keepaspectratio]{mshell}}]{Michael Shell}
% or if you just want to reserve a space for a photo:

% You can push biographies down or up by placing
% a \vfill before or after them. The appropriate
% use of \vfill depends on what kind of text is
% on the last page and whether or not the columns
% are being equalized.

%\vfill

% Can be used to pull up biographies so that the bottom of the last one
% is flush with the other column.
%\enlargethispage{-5in}

% that's all folks